\documentclass[draft]{IEEEtran}
\IEEEoverridecommandlockouts
\onecolumn

\usepackage{amsmath,amssymb,amsthm}
\usepackage[final,hyperfootnotes=false]{hyperref}
\usepackage{dblfloatfix}
\usepackage{tikz}
\usetikzlibrary{arrows,positioning,fit,backgrounds,shapes}

\DeclareMathOperator*{\argmax}{arg\,max}

\newtheorem{thm}{\bf Theorem}
\newtheorem*{hmt}{\bf Hall's Marriage Theorem}
\newtheorem{lem}{\bf Lemma}

\newtheorem{cor}{\bf Corollary}

\newtheorem{rem}{\bf Remark}

\DeclareMathOperator*{\rmd}{d}
\def\im{\imath_{U;V}}
\def\p{\mathbb{P}}
\def\e{\mathbb{E}}

\begin{document}
\title{Sharp Bounds for Mutual Covering
\thanks{Parts of the paper (mostly, Part 1) described in the abstract) were presented at ISIT 2017.}
}
\author{Jingbo Liu\quad Mohammad H.~Yassaee \quad Sergio Verd\'{u}\\
jingbo@mit.edu\quad
yassaee@ipm.ir\quad
verdu@informationtheory.org
}
\date{}

\maketitle
\begin{abstract}
A fundamental tool in network information theory is the covering lemma, 
which lower bounds the probability that there exists a pair of random variables, among a give number of independently generated candidates, falling within a given set. 
We use a weighted sum trick and Talagrand's concentration inequality to prove new  mutual covering bounds.
We identify two interesting applications: 1) When the probability of the set under the given joint distribution is bounded away from 0 and 1, the covering probability converges to 1 \emph{doubly} exponentially fast in the blocklength, which implies that the covering lemma does not induce penalties on the error exponents in the applications to coding theorems.
2) Using Hall's marriage lemma, we show that the maximum difference between the probability of the set under the joint distribution and the covering probability equals half the minimum total variation distance between the joint distribution and any distribution that can be simulated by selecting a pair from the candidates.
Thus we use the mutual covering bound to derive the exact error exponent in the joint distribution simulation problem.
In both applications,  the determination of the exact exponential (or doubly exponential) behavior relies crucially on the sharp concentration inequality used in the proof of the mutual covering lemma.
%
\end{abstract}

\textbf{Keywords:} Shannon theory, network information theory, covering lemmas, one-shot method, distribution simulation, randomness generation, rejection sampling, concentration inequalities, information density.
\section{Introduction}
A cornerstone in the achievability proofs of many problems in rate-distortion theory and  network information theory is the asymptotic \emph{covering lemma} (see e.g.\ \cite[Lemma 9.1]{csiszar}\cite[Lemma~3.3]{NIT}), which gives sufficient conditions for finding among a given number of independently generated random variables a pair that will be jointly typical.
\par
A non-asymptotic single-shot\footnote{A single-shot formulation refers to the non-asymptotic setting where the random variables are not necessarily vectors with i.i.d.\ coordinates.
Of course, the main contribution of the paper lies in the new arguments for bounding the covering and the sampling error, which are independent of the type of formulation we adopt.} version of the \emph{unilateral covering lemma}\footnote{In contrast to the standard literature, we add the qualifier ``unilateral'' to differentiate the standard covering lemma with the mutual covering lemma.}  was formulated by
Verd\'u \cite{Verduallerton}: Given $P_{UV}$, let $U_1,\dots,U_M\sim P_U$, $V\sim P_V$,
and suppose that $V, U_1,\dots,U_M$
 are independent.
If $\mathcal{F}\subseteq\mathcal{U}\times \mathcal{V}$ is such that $P_{UV}(\mathcal{F})$ is large,
then
\begin{align}
\mathbb{P}\left[\bigcup_{m=1}^M \{(U_m,V)\in \mathcal{F}\}\right]
\label{e_vcovering0}
\end{align}
must be large as well, provided that $M$ is sufficiently large.
The key in Verd\'u's single-shot non-asymptotic covering lemma  \cite{Verduallerton} is to express the 
condition on the size of $M$ in terms of a bound on the information spectrum of the pair $(U,V)$,
namely, $M$ is large with respect to $\exp(\eta)$
where $\eta$ is
such that $\mathbb{P} \left[ \imath_{U;V} ( U; V  ) > \eta \right]$ is small.

The conventional asymptotic version can be recovered if we take
$P_{UV}$ to be $P_{\sf UV}^{\otimes n}$, and
$\mathcal{F}$ to be the typical set.
Then, by the central limit theorem, 
for \eqref{e_vcovering0} to be bounded away from 0 and 1,
one must take $M=\exp(nI({\sf U;V})+O(\sqrt{n}))$.
The mutual covering lemma \cite{el1981proof} (see also \cite[Lemma~8.1]{NIT}),
which finds applications such as Marton's coding scheme for broadcast channels and multiple-description coding,
may be thought of as a generalization where one finds a pair $(U_m,V_l)\in\mathcal{F}$ from $U_1,\dots,U_M\sim P_U$, $V_1,\dots,V_L\sim P_V$.

The standard proof of the unilateral covering lemma \cite[Lemma~3.3]{NIT} does not apply immediately to the mutual covering case because of the more complicated dependence structure between the pairs.
The original proof of the mutual covering lemma \cite{el1981proof} uses a ``second-moment method'', an idea widely used in graph theory (see e.g.\ \cite{alon2004probabilistic}).
An alternative approach based on channel resolvability was recently given in \cite{jingbo}, which gives strictly tighter exponential bounds in certain regimes.
A survey of previous covering lemmas is given in Section~\ref{sec_previous}.

In this paper, we prove a stronger mutual covering lemma with new arguments (Lemma~\ref{le:dsmcl}),
and demonstrate its power in old and new applications of the covering lemma.
\emph{Strictly} improving previous mutual covering bounds, the new bound is \emph{sharp} in the following two regimes:

1) ``\emph{Typical case $\mathcal{F}$}'':
For ``regular'' $P_{UV}$ (in particular, for the stationary memoryless cases where $P_{UV}=P_{\sf UV}^{\otimes n}$),
we have\footnote{Unless specified, the bases in $\log$ and $\exp$ are arbitrary but matching.}
\begin{align}
&\quad\sup_{\mathcal{F}\colon P_{UV}(\mathcal{F})\ge \frac{1}{2}}\mathbb{P}\left[\bigcap_{m=1}^M\bigcap_{l=1}^L (U_m,V_l) \notin \mathcal{F}\right]
\nonumber\\
&\approx\max\left\{
\exp_e(-M),\, \exp_e(-L),\, \exp_e\left(-\frac{ML}{I(U;V)}\right)
\right\}.
\label{e_typical}
\end{align}
See \eqref{e14} for the formal statement.
In particular, in the stationary memoryless case with $M_n$ and $L_n$ growing exponentially in the blocklength $n$, and
\begin{align}
\lim_{n\to\infty}\frac{1}{n}\log M_nL_n> I({\sf U;V})
\label{e_fixrate}
\end{align}
fixed, the covering error for a ``typical'' $\mathcal{F}$ (that is, $P_{\sf UV}^{\otimes n}(\mathcal{F})$ is bounded away from 0 and 1) must be doubly exponentially small in $n$.

2) ``\emph{Worst case $\mathcal{F}$}'': Again for ``regular'' $P_{UV}$,
we have 
\begin{align}
&\quad\sup_{\mathcal{F}}\left\{
P_{UV}(\mathcal{F})
-
\mathbb{P}\left[\bigcup_{m=1}^M\bigcup_{l=1}^L (U_m,V_l) \in \mathcal{F}\right]
\right\}
\nonumber\\
&\approx\mathbb{P}[\imath_{U;V}(U;V)>\log ML]
\label{e_worst}
\end{align}
where $(U,V)\sim P_{UV}$,
and the \emph{information density} is $\imath_{U;V}(u;v):=\log\frac{\rmd P_{UV}}{\rmd(P_U\times P_V)}(u;v)$.
See \eqref{e_conv} for the formal statement.
In the stationary memoryless case where \eqref{e_fixrate} holds,
the right side of \eqref{e_worst} vanishes exponentially (rather than doubly exponentially) and captures the correct exponent of the left side.
This implies that the supremum in \eqref{e_worst} is
achieved by some $\mathcal{F}$ such that $P_{\sf UV}^{\otimes n}(\mathcal{F})=o(1)$, since otherwise the result contradicts the double exponential convergence for the ``typical case $\mathcal{F}$''.

One might ask whether the left side of \eqref{e_typical} or the left side \eqref{e_worst} is the ``right'' notion of the covering error.
Interestingly, we demonstrate two applications where the two quantities respectively play a fundamental role:

1) \emph{Achievable error exponent in the broadcast channel}.
The estimate in \eqref{e_typical} is more relevant to the achievability proof of coding theorems, which is the original motivation for the mutual covering lemma in \cite{el1981proof}.
We present a Gallager-type error exponent for the broadcast channel using Marton's coding scheme.
The doubly exponential decay in mutual covering has the implication that the covering error does not contribute to the error exponent in the broadcast channel.
In contrast, previous covering lemmas give exponential bounds on the covering error, leading to strictly worse error exponents for the broadcast channel.

2) \emph{Joint distribution simulation via selection}.
Given $P_{UV}$, let $U_1,\dots,U_M\sim P_U$, $V_1,\dots,V_L\sim P_V$.
Suppose that we want to approximate $P_{UV}$ by the distribution of $(U_{\hat{M}},V_{\hat{L}})$,
where the indices $\hat{M}\in\{1,\dots,M\}$ and $\hat{L}\in\{1,\dots,L\}$ are selected
upon observing all the random variables $U^M$ and $V^L$.
A duality between covering and sampling recently observed in \cite{lv18} shows that the approximation error in total variation is precisely characterized by the left side of \eqref{e_worst}.
Based on this duality observation,
we derive the exact error exponent as well as the second-order rates (for a nonvanishing error) of joint distribution simulation of stationary memoryless sources.
For contrast, we prove that a simple weighted selection rule (similar to the likelihood encoder \cite{y-isit15}\cite{song16}) is fundamentally incapable of yielding the exact error exponent (Section~\ref{app_simple}).

Distribution simulation has proved to be fertile ground for non-asymptotic achievability bounds \cite{y-isit15}\cite{song16}\cite{watanabe}\cite{liu2016key}\cite{jingbo2015egamma}\cite{liu2017secret}\cite{harsha2010}\cite{LiElGamal}\cite{li2018unified}.
For example, in the likelihood encoder approach \cite{song16}, the encoder selects a codeword such that a target joint distribution with the source is achieved, even though they are independently generated.
Non-interactive distribution generation,
a topic in information theory and theoretical computer science \cite{witsenhausen1975}\cite{mossel2006}\cite{kamath},
may appear similar to the joint distribution simulation problem mentioned above.
However, the joint distribution simulation setup we consider is actually quite different,
in the sense of being ``fully interactive'' ($\hat{M}$ is selected upon viewing both $U^M$ and $V^L$, instead of just $U^M$), 
and moreover the output must be one of the samples rather than a function of the observed sequence.

As alluded above, in both the ``worst case'' and the ``typical case'', in order to get the exact expression for the convergence rates, the main technical challenge is to prove new and tight bounds for mutual covering.
The main mutual covering lemma we prove suffices for both cases.
The proof idea is to lower bound the probability that a real valued random variable is positive, just as in \cite[Lemma~8.1]{NIT},
but with two important innovations:
\begin{itemize}
\item The original proof of mutual covering \cite{NIT} bounded the probability that the sum of the \emph{indicator} functions of covering events is positive.
Here we use a \emph{weighted}\footnote{In the asymptotic setting, the weights of the non-zero indicators are almost equal, so the weighted sum is approximately a scaled version of the number of typical events.} sum instead, which results in a more compact (and slightly stronger) bound than the one-shot bound obtained by directly considering the plain sum of indicators \cite{Radhakrishnan2016}. Furthermore, the extension to the multivariate setting is straightforward.

\item While the original proof of mutual covering \cite{NIT} bounded the deviation of the sum of indicator functions from its mean via  Chebyshev's inequality,
we bound the deviation using a one-sided version of the Talagrand inequality \cite[Theorem 8.6]{BLM} in addition to Chebyshev's inequality.
Using the Talagrand inequality,
in the asymptotic setting ($U^{(n)}\sim P_{\sf U}^{\otimes n}$, $V^{(n)}\sim P_{\sf V}^{\otimes n}$) we observe that the probability of failure to cover converges to zero doubly exponentially in $n$ provided that $P_{\sf UV}^{\otimes n}(\mathcal{F})$ is bounded away from zero.
This is a substantial improvement on the convergence rates of previous  bounds.
\end{itemize}

The mutual covering bounds in this paper are easily extendable from the bivariate to the  multivariate case, 
which is useful for $m$-user broadcast channel and multiple descriptions \cite{NIT}.

The remainder of the paper is organized as follows: 
Section~\ref{sec_previous} surveys previously published covering lemmas;
Section~\ref{sec_results} lists the main results of the paper, including new bounds for covering and the applications in network information theory and in sampling.
Section~\ref{sec_ext} discusses some extensions to multivariate and conditional settings.
Section~\ref{sec_proof} onward are devoted to the proofs of the results.

\section{Previous Single-Shot Covering Bounds}
\label{sec_previous}
In this section we survey previous single-shot covering lemmas, starting with the unilateral setting after which we discuss two mutual covering lemmas (for an earlier survey of non-asymptotic covering lemmas, see also \cite{verdu-jerusalem}).

\begin{lem}[Unilateral covering lemma] \cite{Verduallerton}. \label{lem:unilateral}
Let $P_{UV}$ be given. Let $U_1,\dots,U_M\sim P_U$ and $V\sim P_V$ be independent.
For any $\gamma>0$ and $\mathcal{F}\subseteq \mathcal{U}\times\mathcal{V}$,
\begin{align}
&\quad\mathbb{P}\left[\bigcap_{m=1}^M \{(U_m,V)\notin \mathcal{F}\}\right]
\nonumber\\
&\le
P_{UV}(\mathcal{F}^c)
+\mathbb{P}[\imath_{U;V}(U;V)\ge\log M-\gamma]
\nonumber\\
&\quad+e^{-\exp(\gamma)}
\label{e_vcovering}
\end{align}
where $(U,V)\sim P_{UV}$.
\end{lem}
The proof uses the inequality
\begin{align}
\left(1-\frac{p\alpha}{M}\right)^M
\le 1-p+e^{-\alpha}
\end{align}
for $M,\alpha>0$ and $0\le p\le1$, which is a standard tool in the achievability proof in rate distortion theory.
In mutual covering, we no longer have a single random variable which we hope will be ``covered" by at least one of $M$ independent
realizations; instead we have $L$ independent realizations and the failure to cover event is
\[
\bigcap_{m=1}^M\bigcap_{l=1}^L\{(U_m,V_l)
\notin\mathcal{F}\}.
\]
The purpose of the mutual covering lemma is to find an upper bound to the
probability of covering failure as a function of $\mathcal{F}$, and of $P_{UV}(\mathcal{F}^c)$, the
probability of failure under a nominal joint probability measure.
\begin{lem}[Mutual covering lemma] \cite{jingbo} 
Let $P_{UV}$ be given. Let $U_1,\dots,U_M$ and $V_1,\dots,V_L$ be independent.
For any $\delta,\gamma>0$ and $\mathcal{F}\subseteq \mathcal{U}\times\mathcal{V}$,
\begin{align}
&\mathbb{P}
\left[\cap_{m=1}^M\cap_{l=1}^L\{(U_m,V_l)
\notin\mathcal{F}\}\right] \nonumber\\
&\le P_{UV}(\mathcal{F}^c) + \mathbb{P}\left[\exp(\imath_{U;V}(U;V))\ge \tfrac{ML}{\exp(\gamma)}-\delta\right]
\nonumber\\
&\quad+\tfrac{\min\{M,L\}-1}{\delta}
+e^{-\exp(\gamma)}.
\label{e_mybound}
\end{align}
\end{lem}
A more ``structured'' but essentially similar proof based on $E_{\gamma}$-resolvability yields a  bound similar to \eqref{e_mybound}; see
\cite{jingbo}\cite{jingbo2015egamma}.
Taking $M=1$ and $\delta\downarrow 0$ in \eqref{e_mybound}, we recover the one-shot unilateral covering lemma in \cite{Verduallerton}.
A pleasing property of \eqref{e_mybound} is that it is linear in both probability terms, and hence the bound can be applied to the achievability proof of the broadcast channel in a similar manner as \cite{Verduallerton}.
Despite the symmetry of  \eqref{e_mybound}, in the proof the roles of  $U$ and $V$ are asymmetric, hence not easily extended to the multivariate case.

The original idea for the asymptotic mutual covering in \cite{el1981proof} based on calculating the expected number of typical pairs and bounding the deviation by Chebyshev's inequality can be also applied to the one-shot setting as shown in the following result.
\begin{lem}[Mutual covering lemma] \cite{Radhakrishnan2016}.
Given $P_{UV}$, $\mathcal{F}$,
$(U_1,\dots,U_M)$ and $(V_1,\dots,V_L)$ as before,
for any $\epsilon\in (0,1)$ we have
\begin{align}
\mathbb{P}\left[\bigcap_{m=1}^M
\bigcap_{l=1}^L\{(U_m,V_l)
\notin\mathcal{F}\}\right]
&\le\tfrac{\exp(I_{\infty}^{\epsilon})}
{\left(P_{UV}(\mathcal{F}^c)-\epsilon\right)ML}
\nonumber\\
&+
\tfrac{M+L}{\left(P_{UV}(\mathcal{F}^c)-\epsilon\right)^2ML}
\label{e_rad}
\end{align}
where the smooth mutual information (see e.g.\ \cite{renner2005simple}) is defined as
\begin{align}
I_{\infty}^{\epsilon}&:=
\inf_{\mathcal{G}\colon P_{UV}(\mathcal{G})\ge 1-\epsilon}
\sup_{(u,v)\in\mathcal{G}}\imath_{U;V}(u;v).
\label{e_smoothmi}
\end{align}
\end{lem}
Yassaee et al.\ \cite{y-isit13} proposed a general method to one-shot achievabilities based on likelihood encoders/decoders (that sample from the posterior rather than declaring the argument that maximizes it) and Jensen's inequality.
    In \cite{y-isit15}, Yassaee applied this approach to the covering problems.
    In Theorem~\ref{thm7} we reproduce a bound presented by Yassaee in the conference presentation of \cite{y-isit15} (not included in \cite{y-isit15}).


\begin{rem}\normalfont
None of the mutual covering lemmas above is tight enough to show \eqref{e_typical}.
In the stationary memoryless setting \eqref{e_fixrate}, these bounds only show that the left side of \eqref{e_typical} is exponentially small in the block length $n$ (rather than doubly exponential as in \eqref{e_typical}).
\end{rem}
\begin{rem}\normalfont
Lemma \ref{lem:unilateral} establishes \eqref{e_worst} for $M=1$ or $L=1$ (by choosing $\gamma=0.001\log ML$, say).
However, the existing mutual covering lemmas are not sufficient for establishing \eqref{e_worst} for general $M$ and $L$.
In particular, \eqref{e_mybound} establishes \eqref{e_worst} only when $\min\left\{\frac{1}{L},\frac{1}{M}\right\}$ is negligible compared to $\mathbb{P}[\imath_{U;V}(U;V)\ge \log ML]$ (by taking $\gamma=0.001\log ML$ and $\delta=0.5(ML)^{0.999}$, say).
\end{rem}

\section{Summary of the Main Results}
\label{sec_results}
\subsection{New Bounds for Mutual Covering}
The proof of our new mutual covering lemma relies on two tools not used before in this context, namely, a weighted sum trick and a Talagrand concentration inequality.
Instead of examining the full combined power of these tools,
we first give a simpler version assuming a weaker
quadruple-wise independence structure among the random variables,
in a similar vein as \eqref{e_rad}. 
This limited independence setting is useful in the analysis of linear codes in which the codewords are generated from a random matrix and are not mutually independent.
\begin{lem}[Limited independence]\label{le:smcl}
Given $P_{UV}$, assume that $(U_1,\cdots,U_M,V_1,\cdots,V_L)$ satisfy
\begin{equation}
P_{U_mU_{\bar{m}}
V_lV_{\bar{l}}}
=P_U\times P_U\times P_V\times P_V,\quad \forall m\neq\bar{m},~l\neq\bar{l}.
\label{eqn:y0}
\end{equation}
Then for any $\mathcal{F}\subseteq\mathcal{U}\times \mathcal{V}$ and $\gamma>0$,
\begin{align}
\p&\left[\bigcap_{m=1}^M
\bigcap_{l=1}^L
\{(U_m,V_l)\notin\mathcal{F}\}\right]
\nonumber\\
&\le \dfrac{\exp(-\gamma)+M^{-1}+L^{-1}}
{\p\left[(U,V)\in\mathcal{F}, ~\im(U;V)\le\log ML-\gamma\right]}.
\end{align}
\end{lem}
The proof of Lemma~\ref{le:smcl} uses the weighted sum trick
and is given in Section~\ref{sec_proof}. 
\par
Under a stronger independence structure, we replace the Chebyshev inequality in the proof of Lemma~\ref{le:smcl} with the Talagrand inequality, and get a tighter bound sufficient to yield both \eqref{e_typical} and \eqref{e_worst}.  This leads to the following lemma, which is the main result of the paper
and is proved in  Section~\ref{sec_proof}.

\begin{lem}[Full independence]\label{le:dsmcl}
Given $P_{UV}$,
suppose that $(U_1,\cdots,U_M,V_1,\cdots,V_L)$ has the joint distribution
\begin{equation}
P_{U^MV^L}=\underbrace{P_U\times\cdots\times P_U}_{\mbox{\tiny{$M$ times}}}\times \underbrace{P_V\times\cdots\times P_V}_{\mbox{\tiny {$L$ times}}},
\label{eqn:dy0}
\end{equation}
then for any event $\mathcal{F}$ and $\gamma>0$,
\begin{align}
&\p\left[\bigcap_{m=1}^M
\bigcap_{l=1}^L\{(U_m,V_l)\notin\mathcal{F}\}\right]
\nonumber\\
&\le \exp_{e}\left(-\dfrac{\p\left[(U,V)\in\mathcal{F}, ~\im(U;V)\le\log ML-\gamma\right]}{4\exp(-\gamma)+2M^{-1}+2{L}^{-1}}\right).
\end{align}
\end{lem}

\begin{rem}\normalfont
The probability of failure to cover $\mathcal{F}$ can be optimized over all couplings $P_{UV}$ with marginals $P_U$ and $P_V$. This allows us to find a good joint distribution  resulting in vanishing failure probability, which need not be the case
for an arbitrary coupling $P_{UV}$.
\end{rem}

\subsection{Double Exponential Convergence of the Covering Error}
\label{sec_dexp}
With the smooth mutual information defined in \eqref{e_smoothmi},
for any $p\in[0,1]$ we have
\begin{align}
&\quad
\sup_{\mathcal{F}\colon P_{UV}(\mathcal{F})\ge p}
\p\left[\bigcap_{m=1}^M
\bigcap_{l=1}^L\{(U_m,V_l)\notin\mathcal{F}\}\right]
\nonumber\\
&\le \inf_{\epsilon\in[0,p]} \exp_e\left(-\frac{p-\epsilon}{\tfrac{2}{M}+\tfrac{2}{L}
+\tfrac{4\exp(I_{\infty}^{\epsilon}(U;V))}{ML}}\right)
\label{e14}
\end{align}
which follows from Lemma~\ref{le:dsmcl} by taking $\gamma=\log ML-I_{\infty}^{\epsilon}(U;V)$.
Thus in an asymptotic setting where $\frac{1}{n}\imath_{U^{(n)};V^{(n)}}(U^{(n)};V^{(n)})$ satisfies the law of large numbers, the smooth mutual information can be approximated by the mutual information,
and we see that the $\le$ part of \eqref{e_typical} holds.
In particular, in the stationary memoryless case, we can derive the exact doubly exponential exponent (by finding an asymptotically matching converse):
\begin{thm}\label{thm_dee}
Fix\footnote{When we derive single-letter expressions in the stationary memoryless settings from bounds in the single-shot settings, we will use the superscript $(n)$, as in $U^{(n)}$, to indicate an $n$-vector in $\mathcal{U}^n$; san-serif letters, as in $P_{\sf UV}$, are used in per-letter distributions.} $P_{\sf UV}$, $R_1,R_2>0$, $\epsilon\in (0,1)$.
Let $U^{(n)}_1,\dots,U^{(n)}_M\sim P^{\otimes n}_{\sf U}$
and $V^{(n)}_1,\dots,V^{(n)}_L\sim P_{\sf V}^{\otimes n}$ be independent,
where\footnote{Here and elsewhere, as usual the rounding of the right sides of \eqref{e_Mr1} and \eqref{e_nr2} is not explicitly indicated.}
\begin{align}
M_n&=\exp(nR_1);\label{e_Mr1}
\\
L_n&=\exp(nR_2).\label{e_nr2}
\end{align}
Then
\begin{align}
&\quad\lim_{n\to\infty}\frac{1}{n}
\log\log
\frac{1}{\sup_{\mathcal{F}}\mathbb{P}[\bigcap_{m=1}^{M_n}\bigcap_{l=1}^{L_n}
(U^{(n)}_m,V^{(n)}_l)\notin \mathcal{F}]}
\nonumber\\
&=
\min\{R_1,R_2,R_1+R_2-I(\sf{U;V})\}.
\label{e_dee}
\end{align}
where the supremum is over all $\mathcal{F}\subseteq \mathcal{U}^n\times \mathcal{V}^n$ such that
$P_{UV}^{\otimes n}(\mathcal{F})\ge 1-\epsilon$.
\end{thm}
The proof of Theorem~\ref{thm_dee} is given in Section~\ref{sec_doubleexp}.
As is well-known, doubly exponential convergence phenomenon plays a role in the analysis of various information-theoretic problems.
(A very partial list of) examples include: error exponent in lossy compression (benefiting from the doubly exponential convergence in unilateral covering) \cite{csiszar};
second-order rate of Gelfand-Pinsker coding \cite[(60)]{scarlett2015};
error exponent of random channel codes \cite[(30)]{merhav}
exponent in channel synthesis \cite{yagli2018exact} (building on the concentration of the soft-covering error probability, a.k.a.\ strong soft-covering lemma \cite[Theorem~31]{jingbo2015egamma}).

\subsection{Application in Error Exponents of Broadcast Channels}
A classical application of the mutual covering lemma is the achievability proof for the broadcast channel \cite{el1981proof},\cite{NIT}.
In the single-shot setting,
consider a broadcast channel with marginal random transformations $P_{Y|X}$ and $P_{Z|X}$.
An $(M_1,M_2)$-code consists of the following:
\begin{itemize}
\item  Encoder maps an equiprobable message $(W_1,W_2)\in\{1,2,\dots,M_1\}\times\{1,2,\dots,M_2\}$ to an element in $\mathcal{X}$;
\item  Decoder~1 maps an element in $\mathcal{Y}$ to $\hat{W}_1\in\{1,2,\dots,M_1\}$.
Similarly Decoder~2 maps an element in $\mathcal{Z}$ to $\hat{W}_2$.
\end{itemize}
The goal is to design the encoder and decoders to  minimize the error probabilities $\mathbb{P}[\hat{W}_k\neq W_k]$, $k=1,2$.

Switching to the discrete memoryless setting, 
we are given the per-letter random transformations $P_{\sf Y|X}$ and $P_{\sf Z|X}$,
and make the substitutions $P_{Y|X}\leftarrow P_{\sf Y|X}^{\otimes n}$ and $P_{Z|X}\leftarrow P_{\sf Z|X}^{\otimes n}$ in the above definition, where $n$ denotes the blocklength.
A rate pair $(R_1,R_2)\in(0,\infty)^2$ is said to be achievable if there exists a sequence of $(n,M_{1n},M_{2n})$ codes such that 
\begin{align}
\liminf_{n\to\infty}\frac{1}{n}\log M_{kn}\ge R_k,
\quad k=1,2
\end{align}
and 
\begin{align}
 \limsup_{n\to\infty}\max_{k=1,2}\mathbb{P}[\hat{W}_k\neq W_k]
 =0.\label{e19}
\end{align}
A two-auxiliary simplification of \emph{Marton's inner bound} \cite{marton1979} (reproduced in \cite[Theorem~8.3]{NIT}) states that $(R_1,R_2)$ is achievable if there exists a discrete distribution $P_{\sf UV}$ and a function $x\colon\mathcal{U}\times \mathcal{V}\to \mathcal{X}$ such that
\begin{align}
R_1&\le I({\sf U;Y}),\label{e_marton0}
\\
R_2&\le I({\sf V;Z}),
\\
R_1+R_2&\le I({\sf U;Y})+I({\sf V;Z})-I({\sf U;V}).
\label{e_marton}
\end{align}
We briefly recall the proof strategy (informal): Generate $U$ and $V$ codebooks at rates $I({\sf U;Y})$ and $I({\sf V;Z})$ respectively. 
Sacrifice rates $I({\sf U;Y})-R_1$ and $I({\sf V;Z})-R_2$ in the two codebooks so that given any message, a jointly typical pair of codewords can be found, 
in view of \eqref{e_marton} and the mutual covering lemma.
Then a joint typicality rule ensures the $U$-codeword (resp.\ $V$-codeword) to be found at the Decoder~1 (resp.\ Decoder~2).
We remark that \eqref{e_marton0}-\eqref{e_marton} is known to be not tight for certain broadcast channels;
the full Marton bound \cite{marton1979} contains a third auxiliary,
an equivalent form of which was obtained by Liang and Kramer \cite{liangK2007} (reproduced in \cite[Theorem~8.4]{NIT}) by introducing a common message and applying rate splitting.
The proof by Liang and Kramer uses a conditional version of the mutual covering lemma, which we discuss in Section~\ref{sec_ext}.

In this section we investigate the achievable error exponent in the broadcast channel,
that is, a lower bound on $\liminf_{n\to\infty}\frac{1}{n}\log\frac{1}{\max_{k=1,2}\mathbb{P}[\hat{W}_k\neq W_k]}$.
Our result can be seen as an error exponent counterpart of \eqref{e_marton0}-\eqref{e_marton}.
Liang's three-auxiliary inner bound  \cite{liangK2007}\cite[Theorem~8.4]{NIT}
uses a conditional version of mutual covering lemma. We discuss a single-shot version of such a conditional mutual covering lemma later in Lemma~\ref{le:multivariate}.

From the description above of the proof of Marton's inner bound,
one would expect that the error exponent for Decoder~1 comes from taking the minimum between the error exponent in the mutual covering part and the error exponent for the channel from $\mathcal{U}$ to $\mathcal{Y}$;
one might naively imagine that former plays a nontrivial role.
The surprising double exponential convergence phenomenon in Section~\ref{sec_dexp} shows that this is not the case.
This strictly improves the error exponent obtained by  previous mutual covering bounds.

\begin{thm}\label{thm_m_exp}
Consider a discrete memoryless broadcast channel with (per-letter) marginal random transformations $P_{\sf Y|X}$ and $P_{\sf Z|X}$.
Let $P_{\sf UV}$ be an arbitrary discrete distribution,
and $x\colon \mathcal{U}\times \mathcal{V}\to \mathcal{X}$ an arbitrary function.
Let $(R_1,R_2)\in [0,\infty)^2$. 
Then there exists a code at rates $(R_1,R_2)$ achieving the following error exponent
\begin{align}
&\quad\liminf_{n\to\infty}\frac{1}{n}\log \frac{1}{\max_{k=1,2}\mathbb{P}[\hat{W}_k\neq W_k]}
\nonumber\\
&\ge \max_{\theta \in[0,1]}\min\left\{
\begin{array}{c}
E_0(\theta,P_{\sf U},P_{\sf Y|U})-\theta R_1,\\
E_0(\theta,P_{\sf V},P_{\sf Z|V})-\theta R_2,\\
\frac{E_0(\theta,P_{\sf U},P_{\sf Y|U})+E_0(\theta,P_{\sf V},P_{\sf Z|V})
-\theta(R_1+R_2+I({\sf U;V}))}{2}
\end{array}
\right\}
\label{e23}
\end{align}
where $E_0(\theta,P_{\sf U},P_{\sf Y|U}):=\log\left(\sum_y\left(\sum_uP_{\sf U}(u)P_{\sf Y|U}^{\frac{1}{1+\theta}}(y|u)\right)^{1+\theta}\right)$
denotes Gallager's function \cite{gallager},
and $P_{\sf Y|U}$ is induced by $P_{\sf Y|X}$, $P_{\sf UV}$ and $x(\cdot)$.
\end{thm}
Note that by taking $\theta\to0$,
we have $E_0(\theta,P_{\sf U},P_{\sf Y|U})=\theta I({\sf U;Y})+o(\theta)$,
and \eqref{e23} can be matched to \eqref{e_marton0}-\eqref{e_marton}.
We remark that slightly sharper versions are obtained in the proof of Theorem~\ref{thm_m_exp}; 
we present a slightly weakened but simpler version in \eqref{e23}.

\subsection{Joint Distribution Simulation via Covering}
Next, we investigate a new application of mutual covering in the setup of \emph{joint distribution simulation}, in which we desire to select among $ML$ pairs of samples, generated according
to some arbitrary distribution, a pair distributed according to (or closely in total variation distance) a pre-specified  joint distribution. 
The simple relation between covering and distribution simulation, 
observed in \cite{lv18},
leads us to the following covering-sampling duality result, which shows the existence of a sampling scheme without explicit constructions, using bounds on tail probabilities, such as covering lemmas and concentration inequalities.
\begin{thm}\label{thm_dual}
Suppose that $|\mathcal{U}\|\mathcal{V}|<\infty$,
and let $P_{UV}$ be given.
Consider any $U_1,\dots,U_M$
and $V_1,\dots,V_L$,
not necessarily independent and can have any marginal distribution.
Then
\begin{align}
&\quad\sup_{\mathcal{F}}\left\{
P_{UV}(\mathcal{F})
-
\mathbb{P}\left[\bigcup_{m=1}^M\bigcup_{l=1}^L (U_m,V_l) \in \mathcal{F}\right]
\right\}
\nonumber\\
&=
\frac{1}{2}\inf_{P_{\hat{M}\hat{L}|U^MV^L}}\left|P_{UV}-P_{U_{\hat{M}}V_{\hat{L}}}\right|
\label{e_thm_worst}
\end{align}
where $\hat{M}\in \{1,\dots,M\}$ and $\hat{L}\in\{1,\dots,L\}$ are indices selected upon observing $U^M$ and $V^L$.
\end{thm}
We give a simple proof of Theorem~\ref{thm_dual} in
Section~\ref{sec_dual}, which differs than the proof in \cite{lv18} based on linear programming duality.
Of course, these proofs are essentially minimax duality results and are related.
\begin{rem}\normalfont
Although the cardinality of the alphabets does not appear in \eqref{e_thm_worst}, 
we make such an assumption because the proof given
in Section~\ref{sec_dual}
uses Hall's marriage theorem which requires finite alphabets (and is known to fail for some general alphabets). 
We also remark that $P_{\hat{M}\hat{L}|U^MV^L}$  can be constructed using off-the-shelf bipartite graph matching algorithms (see e.g.\ \cite[Section~13.3]{combinatorics}).
\end{rem}

\subsection{Error Exponent and Second-Order Rate in Joint Distribution Simulation}\label{sec_exp}
In the special case of the joint distribution simulation setup
in which the samples are independent and generated according
to the desired marginals, we can leverage  Lemma~\ref{le:dsmcl} and Theorem~\ref{thm_dual},
to show the following information spectrum bound for the achievability of joint distribution simulation.
\begin{thm}\label{thm_sim}
Given $P_{UV}$, let $U_1,\dots,U_M\sim P_U$, any $p\in(0,1)$, and $V_1,\dots,V_L\sim P_V$ be independent. 
If $M,L\ge \frac{3}{p}\ln\frac{1}{1-p}$ then
\begin{align}
&\quad\inf_{P_{\hat{M}\hat{L}|U^MV^L}}
\frac{1}{2}|P_{U_{\hat{M}}V_{\hat{L}}}-P_{UV}|\le
\nonumber\\
&\max\left\{\mathbb{P}\left[\imath_{U;V}(U;V)>\log \frac{pML}{3\ln\frac{1}{1-p}}\right],\, \right.
\nonumber\\
&\quad\left. \inf_{\epsilon\in[0,p]} \exp_e\left(-\frac{p-\epsilon}{\tfrac{2}{M}+\tfrac{2}{L}
+\tfrac{4\exp(I_{\infty}^{\epsilon}(U;V))}{ML}}\right)\right\}.
\label{e_conv}
\end{align}
\end{thm}
The proof of Theorem~\ref{thm_sim} is given in Section~\ref{sec_sim}.
Ignoring the nuisance parameters, when $M$ and $L$ are large, the right side of \eqref{e_conv} is essentially the information spectrum tail
\begin{align}
\mathbb{P}[\imath_{U;V}(U;V)>\log ML].
\label{e_real}
\end{align}
For example,
for stationary memoryless $P_{UV}$,
the first term in the $\max$ dictates the exponential behavior of the right side of \eqref{e_conv}. 

Our proof via the duality result Theorem~\ref{thm_dual} is a novel contribution,
since other more straightforward sampling schemes do not seem to achieve the exact exponent.
For example, inspired by the likelihood encoder, \cite{y-isit13},  \cite{y-isit15} proposed the a simple sampling scheme, which we refer to as the \emph{weighted sampler}.
In Section~\ref{app_simple} however, we show that an analysis of the weighted sampler using the Jensen's inequality argument of \cite{y-isit13} gives an achievability bound of
\begin{align}
\inf_{\gamma>0}\{\mathbb{P}[\imath_{U;V}(U;V)>\log ML-\gamma]+\exp(-\gamma)\}
\label{e_simple}
\end{align}
which is good enough to show that mutual information is the (first-order) fundamental limit
but is exponentially worse than \eqref{e_real}.
Such exponential looseness cannot be overcome with a tighter bound: 
we can in fact determine the exact exponent for the weighted sampler in a certain regime, 
which is strictly worse than the exponent achieved by an optimal sampler in Theorem~\ref{thm_sim};
see Section~\ref{app_simple} for details on the analysis of the weighted sampler.

The following result is a converse  counterpart to Theorem~\ref{thm_sim}.
\begin{thm}\label{thm_conv}
Given $P_{UV}$, let $U_1,\dots,U_M\sim P_U$ and $V_1,\dots,V_L\sim P_V$, all independent.
Then for any random variables $\hat{M}\in\{1,\dots,M\}$ and $\hat{L}\in\{1,\dots,L\}$,
\begin{align}
&\quad\frac{1}{2}|P_{U_{\hat{M}}V_{\hat{L}}}-P_{UV}|
\nonumber\\
&\ge\sup_{\lambda>0}\{(1-\exp(-\lambda))\mathbb{P}[\imath_{U;V}(U;V)>\log ML+\lambda]\}.
\label{e_conv-2}
\end{align}
\end{thm}
The proof of Theorem~\ref{thm_conv} is given in Section~\ref{sec_simconv},
which also uses the duality result in Theorem~\ref{thm_dual}.

In the stationary memoryless case,
Theorem~\ref{thm_sim} and Theorem~\ref{thm_conv}
imply that
\eqref{e_real} is a tight approximation of the error in terms of the error exponent or the second-order analysis.
In particular, we obtain the exact error exponent and the second-order rates in the nonvanishing error regime in the following corollaries.
\begin{cor}\label{cor_simulation}
Fix a discrete memoryless source $P_{\sf UV}$, $R_1,R_2>0$. Let $U^{(n)}_1,\dots,U^{(n)}_M\sim P^{\otimes n}_{\sf U}$
and $V^{(n)}_1,\dots,V^{(n)}_L\sim P_{\sf V}^{\otimes n}$ be independent,
where
\begin{align}
M_n&=\exp(nR_1);\label{e_M}
\\
L_n&=\exp(nR_2).\label{e_L}
\end{align}
Then
\begin{align}
&\quad-\lim_{n\to\infty}\frac{1}{n}\log\inf_{P_{\hat{M}\hat{L}|U^{(n)}_1\dots U^{(n)}_{M_n}V^{(n)}_1\dots V^{(n)}_{L_n}}}
|P_{\sf UV}^{\otimes n}-P_{U^{(n)}_{\hat{M}}V^{(n)}_{\hat{L}}}|
\nonumber\\
&=\max_{\alpha\in (0,\infty)}
\{
\alpha(R_1+R_2)-
\alpha D_{1+\alpha}(P_{\sf UV}\|P_{\sf U}\times P_{\sf V})\},\label{e31}
\end{align}
where 
$$D_{\alpha}(P\|Q):=\frac{1}{\alpha-1}\log\int \left(\frac{{\rm d}P}{{\rm d}Q}\right)^{\alpha}{\rm d}Q $$
denotes the R\'enyi divergence of order $\alpha$ (see e.g.\ \cite{book}). 
\end{cor}
\begin{proof}[Proof sketch]
From the Chernoff bound we obtain 
\begin{align}
&\quad\frac{1}{n}\log \mathbb{P}[\imath_{U^{(n)};V^{(n)}}(U^{(n)};V^{(n)})>\log M_nL_n]
\nonumber\\
&\le \frac{1}{n}\log \mathbb{E}[\exp(\alpha\imath_{U^{(n)};V^{(n)}}(U^{(n)};V^{(n)}))]
-\alpha\log M_nL_n
\nonumber
\\
&=  \alpha D_{1+\alpha}(P_{\sf UV}\| P_{\sf U}\times P_{\sf V}) -\alpha (R_1+R_2)
\nonumber
\end{align}
where $P_{U^{(n)}V^{(n)}}:=P_{\sf UV}^{\otimes n}$.
Moreover it is well-known in large deviation theory that this bound is asymptotically tight upon optimizing $\alpha$.
The claim then follows from Theorem~\ref{thm_sim},
since choosing any $0<\epsilon<p<1$ independent of $n$, we see that  \eqref{e_real} is an exponentially tight approximation.
\end{proof}

\begin{cor}\label{cor_simulation1}
Fix a discrete memoryless source $P_{\sf UV}$, and $R_1,R_2>0$. Let $U^{(n)}_1,\dots,U^{(n)}_{M_n}\sim P^{\otimes n}_{\sf U}$
and $V^{(n)}_1,\dots,V^{(n)}_{L_n}\sim P_{\sf V}^{\otimes n}$ be independent,
where $M_n$ and $L_n$ are positive integers such that
\begin{align}
\lim_{n\to\infty}
\frac{\log M_nL_n-nI({\sf U;V})}{\sqrt{n}}
=
A
\end{align}
for some $A\in\mathbb{R}$. Then
\begin{align}
\lim_{n\to\infty}\inf_{P_{\hat{M}\hat{L}|U^{(n)}_1\dots U^{(n)}_{M_n}V^{(n)}_1\dots V^{(n)}_{L_n}}}
|P_{\sf UV}^{\otimes n}-P_{U^{(n)}_{\hat{M}}V^{(n)}_{\hat{L}}}|
={Q}\left(\frac{A}{V ({\sf U;V} )}\right)
\end{align}
where ${Q}(\cdot)$ denotes the complementary Gaussian CDF function:
\begin{align}
{Q}(x):=\int_{x}^{\infty}\frac{1}{\sqrt{2\pi}}
e^{-\frac{t^2}{2}}{\rm d}t,\quad x\in\mathbb{R},
\end{align}
and the \textit{mutual varentropy} is
\begin{align}
V ({\sf U;V} ) =  \sqrt{{\rm Var}(\imath_{\sf U;V}({\sf U;V}))}.
\end{align}
\end{cor}
\begin{proof}[Proof sketch]
Pick any $\epsilon>Q\left(\frac{A}{V({\sf U;V})}\right)$ and $p\in (\epsilon,1)$, and apply Theorem~\ref{thm_sim}.
Observe that $\log M_nL_n-I_{\infty}^{\epsilon}(U^{(n)};V^{(n)})=\Omega(\sqrt{n})$ by the central limit theorem. 
Hence the first term in the max in
\eqref{e_conv} is the dominant term (the second term vanishes faster than exponentially).
The claim follows by applying the central limit theorem to the first term in the max in \eqref{e_conv}.
\end{proof}

\section{Extension to Multivariate Settings}
\label{sec_ext}
The asymptotic multivariate covering lemma \cite[Lemma~8.2]{NIT} is a generalization of the mutual covering counterpart to the case of $k\ge 3$ codebooks.
Given (single-letter distributions) $P_{{\sf V}_1\dots {\sf V}_k}$, suppose the rate of the $V_k$-codebook is $R_k$, then a typical tuple occurs with high probability if for every $\mathcal{S}\subseteq \{1,2,\dots,k\}$,
\begin{align}\label{eq:mul-mut}
\sum_{k\in \mathcal{S}} R_k
&\ge \sum_{k\in \mathcal{S}} H({\sf V}_k)-
H(({\sf V}_k)_{k\in\mathcal{S}}).
\end{align}
Just as in the asymptotic case,
the one-shot mutual covering lemmas (Lemma~\ref{le:smcl} and Lemma~\ref{le:dsmcl}) can be extended to the multivariate setting by the same arguments.
In some applications  such as multiple description coding \cite[Chapter 13]{NIT} and broadcast channel with a common message  \cite[Theorem~8.4]{NIT},
a conditional version of \eqref{eq:mul-mut} is useful,
which states that given $P_{{\sf ZV}_1\dots {\sf V}_k}$, suppose the rate of the $V_k$-codebook is $R_k$ where the codewords are generated independently conditioned on $Z$, 
then a typical tuple occurs with high probability if
\begin{align}
\sum_{k\in \mathcal{S}} R_k
&\ge \sum_{k\in \mathcal{S}} H({\sf V}_k|{\sf Z})-
H(({\sf V}_k)_{k\in\mathcal{S}}|{\sf Z});
\label{eq:c-mul-mut}
\end{align}
see \cite[Lemma~8.2]{NIT}.
We prove the following single-shot multivariate covering lemma, which implies \eqref{eq:c-mul-mut} in the stationary memoryless setting. 
 \begin{lem}\label{le:multivariate}
Let 
 $P_{ZV_{1}\cdots V_{k}}$ be a given joint distribution.
Let $Z\sim P_Z$,
and conditioned on $Z=z$, let $(V_1^{M_1},\cdots,V_k^{M_k})$ be mutually independent, and
 \begin{equation}
 P_{V_i^{M_i}}=\underbrace{P_{V_i|Z=z}\times\cdots\times P_{V_i|Z=z}}_{\small\mbox{$M_i$ times}},
 \quad\forall i=1,2,\dots,k.
 \end{equation}
then for any $\mathcal{F}\subseteq\mathcal{Z}\times\mathcal{V}_1\times\cdots\times\mathcal{V}_k$
and any $\gamma\in\mathbb{R}$,
\begin{align}
\p&\left[\bigcap_{m_1=1}^{M_1}\cdots\bigcap_{m_k=1}^{M_k} \{(Z,V_{1,m_1},\cdots,V_{k,m_k})\notin\mathcal{F}\}\right]\nonumber\\
                                  &\le \e\left[\exp_{e}\left(-\dfrac{\p\Big[(Z,V_1,\cdots,V_k)\in\mathcal{G}\Big|Z\Big]}{k2^{k}\exp(-\gamma)}\right)\right].
                                  \label{eq:er-mul}
\end{align}
where
\begin{equation}
\mathcal{G}:=\mathcal{F}\cap\left(\bigcap_{\mathcal{S}\subseteq \{1,\dots,k\}} \left\{\imath(Z,V_{\mathcal{S}})<\sum_{t\in\mathcal{S}}\log M_t-\gamma\right\}\right),
\label{eq:dfn-G}\end{equation}
in which $V_{\mathcal{S}}:=(V_i:i\in\mathcal{S})$, and $\imath(Z,V_{\mathcal{S}})$ is a shorthand for\footnote{For probability measures $P$ and $Q$ on the same probability space, the relative information is $\imath_{P\|Q}(.):=\log \dfrac{\mathsf{d}P}{\mathsf{d}Q}(.)$.}
\begin{equation}
\imath_{P_{ZV_{\mathcal{S}}}\|P_Z\prod_{i\in\mathcal{S}}P_{V_i|Z}}(Z,V_{\mathcal{S}}),
\label{eq:dfn-I}\end{equation}
\end{lem}
The proof of Lemma \ref{le:multivariate} is relegated to Section~\ref{apx:multivariate}.
%
\begin{rem}\normalfont
Assume that $Z$ is constant in Lemma \ref{le:multivariate}. Using  steps similar to those in the proof of Theorem~\ref{thm_dee}, it can be shown that the error probability \eqref{eq:er-mul} decays double exponentially in the blocklength. More precisely, assuming $M_{in}=\exp(nR_i)$ and setting $P_{V_i}\leftarrow P^{\otimes n}_{{\sf V}_i}$ and $P_{V_1\cdots V_k}\leftarrow P^{\otimes n}_{{\sf V}_1\cdots {\sf V}_k}$ and denoting the left side of \eqref{eq:er-mul} by $P_e(\mathcal{F})$, its exact doubly exponential decay is
given by
\begin{align}
\quad\lim_{n\to\infty}\frac{1}{n}
&\log\log
\frac{1}{\sup_{\mathcal{F}}P_e(\mathcal{F})}
\nonumber\\
&=
\min_{\mathcal{S}\subseteq\{1,\dots,k\}} \left\{ \sum_{i\in\mathcal{S}}R_i-D\left(P_{{\sf V}_{\mathcal{S}}}\bigg\|\prod_{i\in\mathcal{S}}P_{{\sf V}_i}\right) \right\}
\label{e_dee-mul}
\end{align}
where the supremum is over all $\mathcal{F}$ with $P^{\otimes n}_{{\sf V}_1\cdots {\sf V}_k}(\mathcal{F})\ge 1-\epsilon$,
where $\epsilon\in(0,1)$ is independent of $n$.
Note that the expression in $\min$ equals the inequality gap in \eqref{eq:mul-mut}.
In particular, this shows that if strict inequality is achieved in \eqref{eq:mul-mut} then the probability of non-covering is \emph{doubly exponentially} small.
However, the same is not true for general (not constant) $Z$: in the discrete memoryless setting with $P_{ZV_1\dots V_k}\leftarrow P_{{\sf ZV}_1\dots {\sf V}_k}^{\otimes n}$,
we see that $P_e(\mathcal{F})$ decreases \emph{exponentially} under \eqref{eq:c-mul-mut}, with the exponent $\min_{Q_{\sf Z}}D(Q_{\sf Z}\|P_{\sf Z})$ where the minimum is subject to the constraint
\begin{align}
 \sum_{i\in\mathcal{S}}R_i
 \le D\left(P_{{\sf V}_{\mathcal{S}}|{\sf Z}}\left\|\prod_{i\in \mathcal{S}}P_{{\sf V}_i|{\sf Z}}\right|Q_{\sf Z}\right),
 \quad \forall \mathcal{S}\subseteq \{1,2,\dots,k\}.
\end{align}
\end{rem}

\section{Proofs of Single-Shot Covering Lemmas}
\label{sec_proof}
\subsection{Proof of Lemma~\ref{le:smcl}}
The blueprint follows  El Gamal-van~der Meulen's proof of mutual covering lemma in the asymptotic case \cite{el1981proof} with a slight but important change. Fix the subset of pairs $\mathcal{F}$ and
let
\begin{align}
\mathcal{T}&:=\left\{(u,v)\colon \im(u;v)<\log ML-\gamma\right\},
\\
\mathcal{G}&:=\mathcal{F}\cap\mathcal{T}, \\
S&:=\dfrac{1}{ML}\sum_{m=1}^M \sum_{l=1}^L \exp(\im(U_m;V_l))
1\left\{(U_m,V_l)\in\mathcal{G}\right\}\label{eq:eq3}
\end{align}
a random variable which can be viewed as  a \emph{weighted} sum of the indicator functions
that the pairs belong to $\mathcal{G}$.
In contrast, the proof in \cite{el1981proof} uses non-weighted sum representing the total number of times the pairs
belong to $\mathcal{G}$.
We have
\begin{align}
\p\left[\bigcap_{m=1}^M\bigcap_{l=1}^L
\{(U_m,V_l)\notin\mathcal{F}\}\right]
&\le \p[S=0]\label{eq:main-pr}\\
&\le \p\left[\left|S-\e[S]\right|\ge\e[S]\right]\\
&\le \dfrac{\mathsf{Var}[S]}{\e^2[S]}\label{eq:eq7},
\end{align}
where
\begin{equation}
\e[S]=\p\left[(U,V)\in\mathcal{G}\right]\label{expect1}.
\end{equation}
The remainder of the proof is devoted to showing that the variance  of $S$ satisfies
\begin{align} \label{thegoal}
\mathsf{Var}[S] \leq \e[S] \left( \frac1M + \frac1L + \exp (-\gamma) \right) 
\end{align}
To that end, it is convenient to denote
\begin{equation}
Y_{m,l}:=\dfrac{1}{ML}\exp(\im(U_m;V_l))
1\left\{(U_m,V_l)\in\mathcal{G}\right\}\label{eq:def-y}
\end{equation}
In view of the i.i.d. assumption in \eqref{eqn:y0}, we can write
\begin{align}
\mathsf{Var}[S]&=\mathsf{Var}\left[\sum_{m=1,n=1}^{M,L} Y_{m,l}\right]\\
&=\sum_{m=1,l=1}^{M,L}
\sum_{\bar{m}=1,\bar{l}=1}^{M,L}
\mathsf{Cov}\left[Y_{m,l},Y_{\bar{m},\bar{l}}\right]\\
&=ML\mathsf{Var}[Y_{1,1}]\nonumber\\
&\quad+ML(L-1)\mathsf{Cov}[Y_{1,1},Y_{1,2}]\nonumber\\
&\quad+ML(M-1)\mathsf{Cov}[Y_{1,1},Y_{2,1}]\nonumber\\
&\quad+ML(M-1)(L-1)\mathsf{Cov}[Y_{1,1},Y_{2,2}]\label{eqn:y1}.
\end{align}
Note that here and after, the computations of the second-order statistics (variances) only used the partial independence \eqref{eqn:y0}, and full independence is not necessary.
Next, we bound the four terms in the right side of \eqref{eqn:y1} separately. First, we have
\begin{align}
&\mathsf{Var}[Y_{1,1}] \le\e[Y_{1,1}^2]\\
&=\dfrac{1}{(ML)^2}\e[\exp(2\im(U_1;V_1))
1\left\{(U_1,V_1)\in\mathcal{G}\right\}]\\
&=\dfrac{1}{(ML)^2}
\e\left[\exp(\im(U;V))
1\left\{(U,V)\in\mathcal{G}\right\}]\right.\label{eqn:y2}\\
&\le\dfrac{\exp(-\gamma)}{ML}
\p\left[(U,V)\in\mathcal{G}\right],\label{eqn:y3}
\end{align}
where \eqref{eqn:y2} is due to the change of measure and \eqref{eqn:y3} follows from the definition of $\mathcal{G}$.

Next, consider
\begin{align}
\mathsf{Cov}&[Y_{1,1},Y_{1,2}]\le\e[Y_{1,1}Y_{1,2}]\\
                                  &=\dfrac{1}{(ML)^2}\e[\exp(\im(U_1;V_1))1\left\{(U_1,V_1)\in\mathcal{G}\right\}\nonumber\\
                                  &\qquad\qquad\qquad\cdot\exp(\im(U_1;V_2))1\left\{(U_1,V_2)\in\mathcal{G}\right\}]\\
                                  &\le\dfrac{1}{(ML)^2}\e[\exp(\im(U_1;V_1))1\left\{(U_1,V_1)\in\mathcal{G}\right\}\nonumber\\
                                  &\qquad\qquad\qquad\cdot\exp(\im(U_1;V_2))]\\
                                  &=\dfrac{1}{(ML)^2}\e\bigg[\e[\exp(\im(U_1;V_1))1\left\{(U_1,V_1)\in\mathcal{G}\right\}|U_1]\nonumber\\
                                  &\qquad\qquad\qquad\cdot\e[\exp(\im(U_1;V_2))|U_1]\bigg]\label{eqn:y4}\\
                                  &=\dfrac{1}{(ML)^2}\e[\exp(\im(U_1;V_1))1\left\{(U_1,V_1)\in\mathcal{G}\right\}]\label{eqn:y5}\\
                                  &=\dfrac{1}{(ML)^2}\p[(U,V)\in\mathcal{G}]\label{eqn:y6}
\end{align}
where \eqref{eqn:y4} follows from \eqref{eqn:y0}, \eqref{eqn:y5} follows from 
\begin{equation}\label{expect2}
\e[\exp(\im(U_1;V_2))|U_1]=1
\end{equation}
 and \eqref{eqn:y6} follows by change of measure.
Similarly, we have
\begin{equation}
\mathsf{Cov}[Y_{1,1},Y_{2,1}]
\le\dfrac{1}{(ML)^2}\p[(U,V)\in\mathcal{G}].\label{eqn:y7}
\end{equation}

Finally, \eqref{eqn:y0} implies
\begin{equation}
\mathsf{Cov}[Y_{1,1},Y_{2,2}]=0\label{eqn:y8}
\end{equation}
Substituting \eqref{eqn:y3}, \eqref{eqn:y6}--\eqref{eqn:y8} into \eqref{eqn:y1} yields \eqref{thegoal}.
\begin{rem}[Rationale for the weighted sum]\normalfont
To use concentration inequalities to bound the probability of a non-negative random variable being zero, one can expect that the mean  should be bounded away from zero. If we use a non-weighted sum (i.e. $S':=\frac{1}{ML}\sum_{m=1}^M\sum_{l=1}^L 1\left\{(U_m,V_l)\in\mathcal{F}\right\}$), then $\e[S']= (P_U\times P_V) (\mathcal{F})$ which is usually small in comparison to \eqref{expect1} for a good coupling $P_{UV}$.
\end{rem}

 \subsection{Proof of Lemma~\ref{le:dsmcl}}
In high dimensional probability,  
we often encounter situations where 
limited independence only gives a polynomial decay of the tail probability (e.g.\ Chebyshev's inequality),
whereas a full independence condition yields an exponential upgrade of the decay of the tail probability (e.g.\ sub-Gaussian concentration).
This phenomenon also happens in our problem. 
We derive a sub-Gaussian bound for the probability in the right side of \eqref{eq:main-pr} by means of the following concentration inequality. 
Such concentration inequalities with one-side Lipschitz conditions were originally due to Talagrand \cite{talagrand1988}\cite{talagrand1995}.
However, a ``modern'' proof using Marton's transportation method can be found in \cite[Theorem 8.6]{BLM}.
\begin{thm}[Talagrand]\label{thm_talagrand}
Let $X_1,\cdots,X_n$ be independent and assume that $f\colon \mathcal{X}^n \to \mathbb{R} $ is  such that
\begin{equation}\label{theTcondition}
f(x^n)-f(y^n)\le\sum_{i=1}^n c_i(x^n)1\{x_i\neq y_i\},~~\forall (x^n,y^n)
\end{equation}
for some functions $c_i\colon\mathcal{X}^n \to \mathbb{R}$. Then, for all $t\ge 0$,
\begin{equation}
\p\left[f(X^n)\le \e\left[f(X^n)\right]-t\right]\le \exp_{e}\left(-\dfrac{t^2}{2 \sum_i \e\left[ c_i^2(X^n)\right]}\right).\label{e_talagrand}
\end{equation}
\end{thm}
\par
Denote
\begin{equation}
y_{m,l}(u_m,v_l):=\dfrac{1}{ML}\exp(\im(u_m;v_l))
1\left\{(u_m,v_l)\in\mathcal{G}\right\}.\label{eq:def-y-det}
\end{equation}
We proceed to specialize Theorem \ref{thm_talagrand} to 
\begin{itemize}
\item
$n = M + L$, \smallskip
\item
$f(u^M,v^L)= \sum_{m=1}^M \sum_{\ell=1}^L y_{m\ell}(u_m,v_l)$,\smallskip
\item
$t=\mathbb{E}[f(U^M,V^L)]= \mathbb{E}[S] = \mathbb{P}[(U,V)\in\mathcal{G}]$,\smallskip
\item
$c_m(u^M,v^L)
=\sum_{\ell=1}^L y_{m,\ell}(u_m,v_l),\quad \forall m\in \{1,\dots,M\},$\smallskip
\item
$
c_{-\ell}(u^M,v^L)=\sum_{m=1}^M y_{m,\ell}(u_m,v_l)
,\quad \forall \ell\in \{1,\dots,L\}.
$
\end{itemize}
\par
We can verify that \eqref{theTcondition} is satisfied since
\begin{align}
f(u^M, v^L)-f(s^M&,w^L)
\le \sum_{m=1}^M c_m(u^M,v^L) 1\{u_m\neq s_m\}\nonumber
\\&+
\sum_{l=1}^L c_{-l}(u^M,v^L) 1\{v_l\neq w_l\}.
\end{align}
Invoking \eqref{eqn:y3} and \eqref{eqn:y6}, we get, for each $m$,
\begin{align}
&\mathbb{E}[c_m^2(U^M,V^L)]
=L\mathbb{E}[Y_{11}^2]
+L(L-1)\mathbb{E}[Y_{11}Y_{12}]
\nonumber\\
&\le M^{-1}\exp(-\gamma)\mathbb{P}[(U,V)\in\mathcal{G}]
+M^{-2} \mathbb{P}[(U,V)\in\mathcal{G}],
\end{align}
Consequently,
\begin{align}
\sum_{m=1}^M \mathbb{E}[c_m^2(U^M,V^L)]
&\le
 (\exp(-\gamma)+M^{-1})\mathbb{P}[(U,V)\in\mathcal{G}].\nonumber
\end{align}
By the same argument,
\begin{align}
\sum_{l=1}^L \mathbb{E}[c_{-l}^2(U^M,V^L)]
&\le
 (\exp(-\gamma)+L^{-1})\mathbb{P}[(U,V)\in\mathcal{G}].
 \end{align}
Finally, Theorem~\ref{thm_talagrand} yields
\begin{align}
\mathbb{P}\left[S=0\right]
&= \mathbb{P}[f(U^M,V^L)\le 0]
\\
&= \mathbb{P}[f(U^M,V^L)\le \mathbb{E}[f(U^M,V^L)]-t]
\\
&\le \exp_{e}\left(
-\frac{\mathbb{P}[(U,V)\in \mathcal{G}]}{4\exp(-\gamma)
+2M^{-1}+2L^{-1}}
\right)
\label{e_bound}
\end{align}
as we wanted to show.
\begin{rem}\normalfont
The McDiarmid inequality (e.g.~\cite[Theorem~3.3.8]{raginsky2014concentration}) has previously been applied to several information-theoretic problems including the strong \emph{soft-covering} problem \cite[Theorem~31]{jingbo2015egamma}.
However, the McDiarmid inequality bounds the \emph{variance proxy} in terms of the $L^{\infty}$ norm of the random variables
$\sum_i \left\|c_i^2(X^n)\right\|_{\infty}$,
hence it is too weak to yield any meaningful mutual covering bound.
The Talagrand inequality improves the $L^{\infty}$ bound to a $L^2$ bound (see the denominator in \eqref{e_talagrand}), while only bounding one side of the tail probability; but the one-sided bound is all we need for our purposes.
The improvement to $L^2$ is only available when one considers the lower tail only.
\end{rem}
Next we show a corollary of 
Lemma~\ref{le:dsmcl},
whose  asymptotic version was originally proved in \cite{pradhan} using Suen's correlation inequality. Consider a bipartite regular graph associated with an $n$-type $P_{\sf UV}$, in which the left vertices  represent the sequences $u^n$ with type $P_{\sf U}$ and the right vertices $v^n$ with type $P_{\sf V}$. 
The left and right sides have $A_n=2^{nH({\sf U})+o(n)}$ and $B_n=2^{nH({\sf V})+o(n)}$ vertices, respectively. 
The vertex $u^n$ is connected to $v^n$, if the pair $(u^n,v^n)$ has  type $P_{\sf UV}$. 
The number of edges is $E_n=2^{nH({\sf UV})+o(n)}$.
\begin{cor}
If we draw equiprobably and independently (with replacement) $M_n$ and $L_n$ vertices from the left and right, respectively,  the probability of drawing no connected pairs is upper bounded by
\[
\exp_{e}\left(-\dfrac{1}{4A_nB_n/E_nM_nL_n
+2/M_n+2/L_n}\right).
\]
\end{cor}
\begin{proof}
Let ${\sf U}^n$ and ${\sf V}^n$ be equiprobably distributed on the set of sequences with a type $P_{\sf UV}$.
The claim follows by applying Lemma~\ref{le:dsmcl} with
$U\leftarrow {\sf U}^n$, $V\leftarrow {\sf V}^n$, 
$\gamma\leftarrow \log\frac{M_nL_nE_n}{A_nB_n}$, $\mathcal{F}$ the set of sequences with type $P_{\sf UV}$, and $P_{UV}$ the equiprobable distribution on the $P_{\sf UV}$-type class.
\end{proof}
We remark that Suen's correlation inequality \cite{suen1990} (see \cite{janson1998} for an improved version by Janson),
while a useful tool in graph theory, 
cannot be used in place of the Talagrand inequality in our proof of the sharp mutual covering lemma:
Suen's inequality provides a lower bound on the probability that the sum of a collection of (possibly correlated) Bernoulli random variables is positive.
In contrast, our proof of the sharp mutual covering lemma (which uses the weighted sum-trick) concerns the sum of real-valued (not necessarily Bernoulli) random variables.

\section{Proof of the Double Exponential Convergence in Covering}
\label{sec_doubleexp}
In this section we give a proof of  Theorem~\ref{thm_dee}. To that end, we split the proof
of  \eqref{e_dee} into the proof of the corresponding inequalities:
\subsection{$\ge$ in \eqref{e_dee}}
Assuming that $\mathcal{F}$ satisfies $P_{\sf UV}^{\otimes n}(\mathcal{F})\ge 1-\epsilon$,
by Lemma~\ref{le:dsmcl} we have
\begin{align}
&\p\left[\bigcap_{m=1}^{M_n}
\bigcap_{l=1}^{L_n}\{(U^{(n)}_m,V^{(n)}_l)\notin\mathcal{F}\}\right]
\nonumber\\
&\le \exp_{e}\left(
-\dfrac{1-\epsilon-\p\left[\imath_{U^{(n)};V^{(n)}}(U^{(n)};V^{(n)})>\log M_nL_n-\gamma\right]}{4\exp(-\gamma)+2M_n^{-1}+2L_n^{-1}}\right),
\label{e_67}
\end{align}
where $P_{U^{(n)}V^{(n)}}=P_{\sf UV}^{\otimes n}$.
In view of assumptions \eqref{e_Mr1} and \eqref{e_nr2}, we can choose 
\begin{align}
\gamma=n\, ( R_1+R_2-I({\sf U;V})-\tau )
\end{align}
where $\tau\in (0,\infty)$ is arbitrary.
Since the numerator in the fraction in \eqref{e_67} converges to $1-\epsilon$, we conclude that
\begin{align}
&\quad\lim_{n\to\infty}\frac{1}{n}
\log\log
\frac{1}{\sup_{\mathcal{F}}\mathbb{P}[\bigcap_{m=1}^{M_n}\bigcap_{l=1}^{L_n}
(U^{(n)}_m,V^{(n)}_l)\notin \mathcal{F}]}
\nonumber\\
&\ge
\min\{R_1,R_2,R_1+R_2-I(\sf{U;V})-\tau\}.
\end{align}
Therefore,  $\ge$ holds in \eqref{e_dee} since $\tau>0$ is arbitrary.
\subsection{$\le$ in \eqref{e_dee}}
To prove this direction, we need to construct $\mathcal{F}$ for which the probability of covering failure
is large enough.
If $\mathcal{\bar{U}}\subseteq \mathcal{U}^n$ satisfies
\begin{align}
P_{\sf U}^{\otimes n}(\mathcal{\bar{U}})\in (1-\epsilon,1-\epsilon/2),
\end{align}
then $\mathcal{F}:=\mathcal{\bar{U}}\times\mathcal{V}^n$ also has large probability
\begin{align}
P_{\sf UV}^{\otimes n}(\mathcal{F})\in (1-\epsilon,1-\epsilon/2),
\end{align}
and 
\begin{align}
\mathbb{P}\left[
\bigcap_{m=1}^{M_n}\bigcap_{l=1}^{L_n} \{(U^{(n)}_m,V^{(n)}_l)\notin \mathcal{F}\}
\right]
&=
(1-P_{\sf U}^{\otimes n}(\mathcal{\bar{U}}))^{M_n}.
\end{align}
Therefore,
\begin{align}
&\quad\lim_{n\to\infty}\frac{1}{n}\log\log\mathbb{P}^{-1}\left[
\bigcap_{m=1}^{M_n}\bigcap_{l=1}^{L_n} \{(U^{(n)}_m,V^{(n)}_l)\notin \mathcal{F}\}
\right]
\nonumber\\
&=\lim_{n\to\infty}\frac{1}{n}\log \left(M_nP_{\sf U}^{\otimes n}(\mathcal{\bar{U}}^c)\right)
\\
&=R_1.
\end{align}
By symmetry, we have also shown that the left side of \eqref{e_dee} is upper bounded by $R_2$.

It remains to show that whenever
\begin{align}
R_1+R_2-I({\sf U;V})<\min\{R_1,R_2\},
\label{e_assump}
\end{align}
the left side of \eqref{e_dee} is upper bounded by $R_1+R_2-I({\sf U;V})$.
For sufficiently large $n$ we can pick $\mathcal{F}:=\{(u^n,v^n)\colon\,
\sqrt{n}A \le \sum_{i=1}^n\imath_{\sf U;V}(u_i;v_i)-nI({\sf U;V})\le \sqrt{n}B \}$ for appropriately chosen $A,B\in\mathbb{R}$ independent of $n$ such that
\begin{align}
P_{\sf UV}^{\otimes n}(\mathcal{F})
&\in (1-\epsilon,1-\epsilon/2);
\\
\mathbb{P}[(U^{(n)}_m,V^{(n)}_l)\in \mathcal{F}]
&=\exp(-nI({\sf U;V})+o(n)).\label{e88}
\end{align}
Then
\begin{align}
&\quad\mathbb{P}\left[
\bigcap_{m=1}^{M_n}\bigcap_{l=1}^{L_n} \{(U^{(n)}_m,V^{(n)}_l)\notin \mathcal{F}\}
\right]
\nonumber\\
&=\mathbb{E}\left[
\mathbb{P}^{L_n}\left[\left.\bigcap_{m=1}^{M_n}
\{(U^{(n)}_m,V^{(n)}_1)\notin \mathcal{F}\}\right|U^{(n)}_1\dots U^{(n)}_{M_n}\right]
\right]
\\
&\ge
\mathbb{P}^{L_n}\left[\bigcap_{m=1}^{M_n}
\{(U^{(n)}_m,V^{(n)}_1)\notin \mathcal{F}\}\right]
\label{e_jensen}
\\
&\ge
\left(
1-M_n\mathbb{P}[(U^{(n)}_1,V^{(n)}_1)\in\mathcal{F}]
\right)^{L_n}
\label{e_ub}
\end{align}
where \eqref{e_jensen} and \eqref{e_ub} follow from Jensen's inequality and the union bound, respectively.
Note that $\lim_{n\to \infty}M_n\mathbb{P}[(U^{(n)}_1,V^{(n)}_1)\in\mathcal{F}]=0$ which follows from  \eqref{e_assump} and \eqref{e88}.
Thus,
\begin{align}
&\quad\limsup_{n\to\infty}\frac{1}{n}\log\log\mathbb{P}^{-1}\left[
\bigcap_{m=1}^{M_n}\bigcap_{l=1}^{L_n} \{(U^{(n)}_m,V^{(n)}_l)\notin \mathcal{F}\}
\right]
\nonumber\\
&\le\lim_{n\to\infty}\frac{1}{n}\log \left(M_nL_n
\mathbb{P}[(U^{(n)}_1,V^{(n)}_1)\in \mathcal{F}]\right)
\\
&=R_1+R_2-I({\sf U;V})
\end{align}
as desired.

\section{Proof of Broadcast Channel Error Exponent}
\label{sec_bc}
In this section we prove Theorem~\ref{thm_m_exp} by first showing a stronger single-shot bound.
\begin{lem}\label{lem_broadcast}
Consider a broadcast channel with marginal random transformations $P_{Y|X}$ and $P_{Z|X}$.
Given any positive integers $M_1$, $M_2$, $N_1$, $N_2$, 
a joint distribution $P_{UV}$, 
functions 
\begin{align}
&x\colon \mathcal{U}\times \mathcal{V}\to \mathcal{X},
\\
&{\tt h}\colon \mathcal{U}\times \mathcal{Y}\to (0,\infty)
\\
&{\tt g}\colon \mathcal{V}\times \mathcal{Z}\to (0,\infty)
\end{align}
and $\gamma\in\mathbb{R}$,
$\theta_1,\theta_2\in [0,1]$,
there exists an $(M_1,M_2)$-code such that 
\begin{align}
 \mathbb{P}[\hat{W}_k\neq W_k]
 &\le 
 \exp\left(-\frac{\mathbb{P}[(U,V)\in\mathcal{F},
 \imath_{U;V}(U;V)\le \log N_1N_2-\gamma]}{4\exp(-\gamma)+2/N_1+2/N_2}\right)
 \nonumber\\
 &\quad+(M_kN_k)^{\theta_k}\exp(\mathbb{E}[\phi_k(U,V)]+\gamma),\quad k=1,2.
 \label{e105}
\end{align}
where 
$\mathcal{F}$ is defined as the set of $(u,v)\in\mathcal{U}\times \mathcal{V}$ satisfying
\begin{align}
\phi_k(u,v)\le \mathbb{E}[\phi_k(U,V)]+\gamma,
\quad k=1,2,\label{e18}
\end{align}
with
\begin{align}
\phi_1(u,v):=
\log\mathbb{E}\left[\left.
 \left(\frac{\mathbb{E}[{\tt h}(\bar{U},Y)|Y]}{{\tt h}(u,Y)}\right)^{\theta_1}
 \right| (U,V)=(u,v)
 \right],
\end{align}
in which
$P_{U,V,Y,\bar{U}}(u,v,y,\bar{u}):=P_{UV}(u,v)P_{Y|X=x(u,v)}(y)P_U(\bar{u})$,
and
\begin{align}
\phi_2(u,v):=
\log\mathbb{E}\left[\left.
 \left(\frac{\mathbb{E}[{\tt g}(\bar{V},Z)|Z]}{{\tt g}(v,Z)}\right)^{\theta_2}
 \right| (U,V)=(u,v)
 \right],\nonumber
\end{align}
in which
$P_{U,V,Z,\bar{V}}(u,v,z,\bar{v}):=P_{UV}(u,v)P_{Z|X=x(u,v)}(z)P_V(\bar{v})$.
\end{lem} 
Note that the function $\phi_k$ depends on $\theta_k$, $k=1,2$. 
The functions $\tt g$ and $\tt h$ are also allowed to depend on $\theta_k$.
There is no confusion in omitting the argument $\theta_k$ since $\theta_k$ is fixed in Lemma~\ref{lem_broadcast}.
\begin{proof}[Proof Lemma~\ref{lem_broadcast}]
Construct codebooks $\mathcal{C}=(U_{mn})_{1\le m\le M_1,\,1\le n\le N_1}$ and $\mathcal{D}=(V_{mn})_{1\le m\le M_2,\,1\le n\le N_2}$ where each codeword is independently generated according to $P_U$ and $P_V$ respectively. 

Encoder: Given $(W_1,W_2)=(m_1,m_2)$, select (if any) $(K_1,K_2)=(n_1,n_2)$ such that $(U_{m_1n_1},V_{m_2n_2})\in\mathcal{F}$.
Then the encoder sends $X=x(U_{m_1n_1},V_{m_2n_2})$.
If there are more than one such pair of $(n_1,n_2)$, the encoder can pick any one of them;
if there is no such pair, the encoder fails and outputs any element in $\mathcal{X}$.

Decoders: Upon receiving $Y$, Decoder~1 outputs
\begin{align}
 (\hat{W}_1,\hat{K}_1)=\argmax_{m,n}{\tt h}(U_{mn},Y)
\end{align}
with arbitrary tiebreaking.
Decoder~2 outputs $(\hat{W}_2,\hat{K}_2)$ using a similar rule.

Error analysis: By the symmetry in the codebook construction, we can assume without loss of generality that $(W_1,W_2)=(1,1)$ and we will compute all probabilities below conditioned on this event.
By Theorem~\ref{le:dsmcl}, 
\begin{align}
 \mathbb{P}[\textrm{encoder fails}]
 &=\mathbb{P}\left[\bigcup_{n_1=1}^{N_1}\bigcup_{n_2=1}^{N_2}\{(U_{1n_1},V_{1n_2})\notin\mathcal{F}\}\right]
 \\
 &\le \exp\left(-\frac{\mathbb{P}[(U,V)\in\mathcal{F},
 \imath_{U;V}(U;V)\le \log N_1N_2]}{4\exp(-\tau)+2/N_1+2/N_2}\right)
 \label{e22}
\end{align}
To analyze the failure probability at Decoder~1, 
note that for any $(u,v)\in\mathcal{F}$ and $y\in\mathcal{Y}$,
\begin{align}
 &\quad\mathbb{P}[\hat{W}_1\neq 1|(U_{1K_1},V_{1K_2},Y)=(u,v,y)]
 \nonumber
 \\
 &\le \mathbb{E}\left[\left.\sum_{m=2}^{M_1}
 \sum_{n=1}^{N_1}
 1\left\{\frac{{\tt h}(U_{mn},y)}{{\tt h}(u,y)}\ge 1\right\}\right|(U_{1K_1},V_{1K_2},Y)=(u,v,y)\right]
 \label{e27}
 \\
 &\le M_1N_1\mathbb{E}\left[
 1\left\{\frac{{\tt h}(\bar{U},y)}{{\tt h}(u,y)}\ge 1\right\}\right]
 \label{e28}
 \\
 &\le M_1N_1\frac{\mathbb{E}[{\tt h}(\bar{U},y)]}{{\tt h}(u,y)}.
 \end{align}
 where \eqref{e27} is the union bound;
 \eqref{e28} follows because from the codebook construction, 
 $U_{mn}\sim P_U$ is independent of $(U_{1K_1},V_{1K_2},Y)$,
for each $m\in\{2,\dots,M_1\}$ and  $n\in\{1,\dots,N_1\}$.
 Since probability does not exceed 1, for $\theta_1\in[0,1]$ we obtain
\begin{align}
\mathbb{P}[\hat{W}_1\neq 1|(U_{1K_1},V_{1K_2},Y)=(u,v,y)]
\le (M_1N_1)^{\theta_1}\left(\frac{\mathbb{E}[{\tt h}(\bar{U},y)]}{{\tt h}(u,y)}\right)^{\theta_1}.
 \label{e24}
\end{align}
Next, note that $x(U_{1K_1},V_{1K_2})$ is the input signal for the broadcast channel.
Integrating both sides of \eqref{e24} with respect to $P_{Y|X=x(u,v)}$ we obtain
\begin{align}
 &\quad\mathbb{P}[\hat{W}_1\neq 1|(U_{1K_1},V_{1K_2})=(u,v)]
 \nonumber\\
 &\le 
 (M_1N_1)^{\theta_1}
 \mathbb{E}\left[\left.
 \left(\frac{\mathbb{E}[{\tt h}(\bar{U},Y)|Y]}{{\tt h}(u,Y)}\right)^{\theta_1}
 \right| (U,V)=(u,v)
 \right]
 \\
 &= (M_1N_1)^{\theta_1}\exp(\phi_1(u,v))
 \\
 &\le (M_1N_1)^{\theta_1}\exp(\mathbb{E}[\phi_1(U,V)]+\tau)
\end{align}
where the last step used \eqref{e18} and the assumption that $(u,v)\in\mathcal{F}$.
Then the proof for Decoder~1 is completed by noting that 
\begin{align}
 \mathbb{P}[\hat{W}_1\neq 1]
 \le \mathbb{P}[\textrm{encoder fails}]+(M_1N_1)^{\theta_1}\exp(\mathbb{E}[\phi_1(U,V)]+\tau)
 \label{e110}
\end{align}
and invoking \eqref{e22}.
The proof for Decoder~2 is similar.
\end{proof}

\begin{proof}[Proof of Theorem~\ref{thm_m_exp}]
The proof follows by specializing the single-shot bound Lemma~\ref{lem_broadcast} to the discrete memoryless setting: Let $S_1,S_2$ be any positive numbers such that 
\begin{align}
 S_1+S_2>I({\sf U;V}).\label{e_ss}
\end{align}
With the substitutions 
\begin{align}
 &\gamma\leftarrow n^{2/3},
 \\
 &M_k\leftarrow \exp(nR_k),
 \\
 & N_k\leftarrow \exp(nS_k),
 \\
 & \phi_k({\sf U,V})\leftarrow \sum_{i=1}^n\phi_k({\sf U}_i,{\sf V}_i),
\end{align}
we see that $\mathcal{F}$ defined by \eqref{e18} satisfies $\lim_{n\to\infty}P_{\sf UV}^{\otimes n}(\mathcal{F})=1$, and hence the first term on the right side of \eqref{e110} vanishes doubly exponentially in $n$,
and 
\begin{align}
 \liminf_{n\to\infty}\frac{1}{n}\log \frac{1}{\mathbb{P}[\hat{W}_k\neq W_k]}
 \ge -\mathbb{E}[\phi_k({\sf U,V})]-\theta_k(R_k+S_k).
\end{align}
Let 
\begin{align}
E:=\liminf_{n\to\infty}\frac{1}{n}\log \frac{1}{\max_{k=1,2}\mathbb{P}[\hat{W}_k\neq W_k]},
\end{align}
and restrict to $\theta_1=\theta_2=\theta$ (this restriction is simply to make the final expression simpler; it is possible to obtain a slightly stronger bound allowing $\theta_1$ and $\theta_2$ to differ).
We see that there exists positive $S_1,S_2$ satisfying \eqref{e_ss} as long as 
\begin{align}
-\mathbb{E}[\phi_k({\sf U,V})]-\theta R_k-E&\ge 0,\quad k=1,2;
\\
\sum_{k=1,2}[-\mathbb{E}[\phi_k({\sf U,V})]-\theta R_k]-2E&\ge \theta I({\sf U ;V}).
\end{align}
Rearranging, we see that the following $E$ is achievable:
\begin{align}
 \max_{\theta\in[0,1]}\min\left\{
\begin{array}{c}
-\mathbb{E}[\phi_1({\sf U,V})]-\theta R_1,\\
-\mathbb{E}[\phi_2({\sf U,V})]-\theta R_2,\\
\frac{-\mathbb{E}[\phi_1({\sf U,V})]
-\mathbb{E}[\phi_2({\sf U,V})]
-\theta(R_1+R_2+I({\sf U;V}))}{2}
\end{array}
\right\}\label{e33}
\end{align}
The proof of the theorem follows by further weakening \eqref{e33}: 
By Jensen's inequality,
\begin{align}
 \mathbb{E}[\phi_1({\sf U,V})]
 &=\mathbb{E}\left[\log\mathbb{E}\left[\left.
 \left(\frac{\mathbb{E}[{\tt h}(\bar{\sf U},{\sf Y})|{\sf Y}]}{{\tt h}(u,{\sf Y})}\right)^{\theta}
 \right| ({\sf U,V})=(u,v)
 \right]\right]
 \\
 &\le\log\mathbb{E}\left[
 \left(\frac{\mathbb{E}[{\tt h}(\bar{\sf U},{\sf Y})|{\sf Y}]}{{\tt h}({\sf U,Y})}\right)^{\theta}
 \right].
 \label{e35}
\end{align}
Furthermore, with ${\tt h}(u,y):=\exp\left(\frac{\imath_{\sf U;Y}(u;y)}{1+\theta}\right)$,
the right side of \eqref{e35} equals $E_0(\theta,P_{\sf U},P_{\sf Y|U})$.
Similar simplifications applies to $\mathbb{E}[\phi_2({\sf U,V})]$,
and the theorem is proved.
 \end{proof}

\section{Proof of the Duality between Covering and Distribution Simulation}\label{sec_dual}
This section proves Theorem~\ref{thm_dual}, invoking a combinatorial result,
Hall's marriage lemma.
In fact, we show a more general claim:
\begin{lem}
Let $\mathcal{Z}$ be finite.
Suppose that $Z_1,\dots,Z_N$ are arbitrarily distributed random variables on $\mathcal{Z}$.
Then, for any $P_Z$ on $\mathcal{Z}$,
\begin{align}
&\quad\sup_{\mathcal{F}}\left\{P_Z(\mathcal{F})
-\mathbb{P}\left[\bigcup_{i=1}^N
\{Z_i\in\mathcal{F}\}\right]
\right\}
\nonumber
\\
&=\frac{1}{2}\inf_{P_{\hat{N}|Z^N}}
\left|
P_{Z_{\hat{N}}}-P_Z
\right|
\label{e_generaldual}
\end{align}
where the random variable $\hat{N}\in\{1,\dots,N\}$.
\end{lem}
\begin{proof}
The $\le$ part is obvious since for any $\mathcal{F}$ and any $P_{\hat{N}|Z^N}$,
\begin{align}
\mathbb{P}\left[\bigcup_{i=1}^N \{Z_i\in \mathcal{F}\}\right]
\ge
\mathbb{P}\left[Z_{\hat{N}} \in \mathcal{F}\right].
\end{align}

The $\ge$ part can be shown via Hall's marriage theorem, reproduced after the proof.
We first prove the result assuming that $P_{Z^N}$ and $P_Z$ are \emph{$K$-type distributions} (i.e.\ all the probabilities in the distribution are multiples of $1/K$) for an arbitrary integer $K$;
the general result will then follow by taking $K\to\infty$ and an approximation argument.

Let $\epsilon$ denote the left side of \eqref{e_generaldual}, which is a multiple of $1/K$.
Construct a bipartite graph with $K$ left vertices and $K(1+\epsilon)$ right vertices.
Each $z^N$ is ``split'' into $K P_{Z^N}(z^N)$ left vertices, that is, we use $K P_{Z^N}(z^N)$ identical vertices to represent $z^N$,
so that all the $K$ left vertices have equal probability $1/K$.
Similarly, $z\in\mathcal{Z}$ is ``split'' into $KP_Z(z)$ right vertices.
Moreover, construct $\epsilon K$ right vertices, each representing a \emph{redundant symbol} which can be set to equal any element in $\mathcal{Z}$.
The redundant symbols will represent the error events where we are allowed to produce any symbol.
There is an edge between $(z_1,\dots,z_N)$ and $z$ if the latter is a coordinate of the former,
and there is also an edge between each redundant symbol and each left vertex.
An example is shown in Fig. 1. In this example, The sequence distribution is  $P_{Z^3}(1,2,3)=\frac{2}{3}$ and $P_{Z^3}(1,1,3)=\frac{1}{3}$; the target distribution $P_Z$ is equiprobable on $\{1,2,3\}$. 
The  lines indicate the permissible selection rules. For example, if the sequence $(x_1, x_2, x_3) = (a,b,b)$, then the output can either be $a$ or $b$, hence there are lines connected to $a$ and $b$. After drawing such a bipartite graph, we solve a graph matching problem, and the solution are the bold lines.

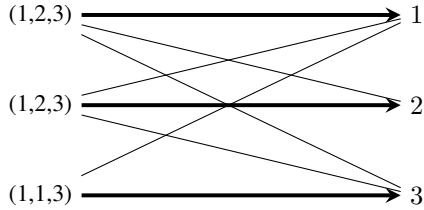
\begin{figure}[h!]
  \centering
\begin{tikzpicture}
[arw/.style={->,>=stealth}]
\node[rectangle] (L1) [xshift=0cm, yshift=2.4cm]{\small(1,2,3)};
\node[rectangle] (L2) [xshift=0cm, yshift=1.2cm] {\small(1,2,3)};
\node[rectangle] (L3) [xshift=0cm, yshift=0cm] {\small(1,1,3)};

\node[rectangle] (R1) [xshift=5cm, yshift=2.4cm]{$1$};
\node[rectangle] (R2) [xshift=5cm, yshift=1.2cm] {$2$};
\node[rectangle] (R3) [xshift=5cm, yshift=0cm] {$3$};

 \draw [arw,line width=1.5pt] (L1) to node {} (R1);
 \draw [arw,line width=1.5pt] (L2) to node {} (R2);
 \draw [arw,line width=1.5pt] (L3) to node {} (R3);

 \draw (L1) to node {} (R2);
 \draw (L1) to node {} (R3);
 \draw (L2) to node {} (R1);
 \draw (L2) to node {} (R3);
 \draw (L3) to node {} (R1);
\end{tikzpicture}
\caption{
The perfect matching in this graph indicates the following rule of selection:
upon observing $(1,2,3)$ the output is 1 or 2 with equal probability;
upon observing $(1,1,3)$ the output is 3.
In this example the approximation error $\epsilon=0$.}
\end{figure}

By Hall's marriage theorem (reproduced after the proof), we claim that there exists a matching that covers the left vertices if for any $\mathcal{S}\subseteq \mathcal{Z}^N$,
\begin{align}
P_{Z^N}(\mathcal{S})
\le
\epsilon+P_Z\left(\bigcup_{z^N\in\mathcal{S}}\{z\colon z\in z^N\}\right).
\label{e_matcond}
\end{align}
Indeed, for any $\mathcal{W}$ a subset of left vertices, 
let $\mathcal{S}$ be the set of all $z^N$ such that $z^N$ corresponds to some vertex in $\mathcal{W}$.
Then $|\mathcal{W}|\le KP_{Z^N}(\mathcal{S})$ (not equality since $\mathcal{W}$ may only contain some, but not all, vertices corresponding to a sequence in $\mathcal{S}$).
The neighborhood of $\mathcal{W}$ in the bipartite graph has cardinality $K\left(\epsilon+P_Z\left(\bigcup_{z^N\in\mathcal{S}}\{z\colon z\in z^N\}\right)\right)$.
Now for any given $\mathcal{S}$, let
\begin{align}
\mathcal{F}:=\bigcap_{z^N\in\mathcal{S}}\{z\colon z\neq z_i,\quad\forall i\}.
\end{align}
We see that
\begin{align}
\mathcal{S}
\subseteq
\bigcap_{i=1}^N\{z^N\colon z_i\notin \mathcal{F}\}
\end{align}
hence \eqref{e_matcond} holds by the definition of $\epsilon$.
The matching then determines a rule of choosing a right vertex (and hence choosing $\hat{Z}$) upon observing $(z_1,\dots,z_N)$, such that
$\frac{1}{2}|P_{Z_{\hat{N}}}-P_Z|$ is bounded by the probability of the redundant right vertices, which is $\epsilon$.

Finally we complete the approximation argument: for general $P_{Z^N}$ and $P_Z$, we can approximate them within any $\delta\in (0,1)$ total variation distance by some $K$-type distributions $P_{\tilde{Z}^N}$ and $P_{\tilde{Z}}$.
\begin{align}
&\quad\sup_{\mathcal{F}}\left\{P_{\tilde{Z}}(\mathcal{F})
-\mathbb{P}\left[\bigcup_{i=1}^N
\{\tilde{Z}_i\in\mathcal{F}\}\right]
\right\}
\nonumber\\
&\le
\epsilon+2\delta.
\end{align}
By the preceding argument we find some $\hat{N}$ such that $\frac{1}{2}|P_{\tilde{Z}_{\hat{N}}}-P_{\hat{Z}}|\le\epsilon+2\delta$.
Applying the same selection rule to $P_{Z^N}$ and using the data processing inequality of the total variation distance, we have
$\frac{1}{2}|P_{Z_{\hat{N}}}-P_{\hat{Z}}|\le\epsilon+3\delta$.
By the triangle inequality,
$
\frac{1}{2}|P_{Z_{\hat{N}}}-P_Z|\le \epsilon+4\delta.
$
The proof is completed since $\delta$ can be made arbitrarily small by choosing large enough $K$.
\end{proof}
\begin{hmt}{\cite{hall}}
Consider a bipartite graph with bipartite sets $\mathcal{X}$ and $\mathcal{Y}$.
There exists a matching that covers $\mathcal{X}$ (i.e.\ an injective map from $\mathcal{X}$ to $\mathcal{Y}$ such that any $x\in\mathcal{X}$ is mapped to an adjacent $y\in\mathcal{Y}$)
if and only if for any $\mathcal{W}\subseteq \mathcal{X}$,
the neighborhood of $\mathcal{W}$ (i.e.\ the set of elements in $\mathcal{Y}$ adjacent to some element in $\mathcal{X}$) has at least the cardinality of $\mathcal{W}$.
\end{hmt}

\section{Proof of Simulation Error Bounds}
This Section proves Theorem~\ref{thm_sim} and Theorem~\ref{thm_conv}.
\subsection{Achievability: Theorem~\ref{thm_sim}}\label{sec_sim}
The achievability follows from the covering-sampling duality (Theorem~\ref{thm_dual}).
In the case of $P_{UV}(\mathcal{F})\ge p$ we can apply \eqref{e14},
so it remains to bound the $P_{UV}(\mathcal{F})< p$ case,
which is accomplished by the following information spectrum bound.
\begin{lem}
Fix $P_{UV}$, $p\in[0,1)$,  and $M,L\ge \frac{6}{p}\ln\frac{1}{1-p}$.
Let $U_1,\dots,U_M\sim P_U$ and $V_1,\dots,V_L\sim P_V$ be independent.
Then
\begin{align}
&\quad\sup_{\mathcal{F}\colon P_{UV}(\mathcal{F})<p}\left\{
\mathbb{P}\left[\bigcap_{m=1}^M\bigcap_{l=1}^L (U_m,V_l)\notin \mathcal{F}\right]
-
P_{UV}(\mathcal{F}^c)
\right\}
\nonumber\\
&\le \mathbb{P}\left[\imath_{U;V}(U;V)>\log \frac{pML}{3\ln\frac{1}{1-p}}\right].
\label{e_sm}
\end{align}
\end{lem}
\begin{proof}
Upon choosing $\gamma=\log \left(\frac{3}{p}\ln\frac{1}{1-p}\right)$, we obtain from Lemma~\ref{le:dsmcl} that
\begin{align}
&\p\left[\bigcap_{m=1}^M
\bigcap_{l=1}^L
\{(U_m,V_l)
\notin\mathcal{F}\}\right]\nonumber\\
&\le \exp_{e}\left(-\frac{1}{p}\ln\left(\frac{1}{1-p}\right)\p\left[(U,V)\in\mathcal{F}, ~\imath_{U;V}(U;V)\le\log \frac{pML}{3\ln\frac{1}{1-p}}\right]\right)
\\
&\le
1-\p\left[(U,V)\in\mathcal{F}, ~\imath_{U;V}(U;V)\le\log \frac{pML}{3\ln\frac{1}{1-p}}\right]
\\
&\le
P_{UV}(\mathcal{F}^c)+\mathbb{P}\left[\imath_{U;V}(U;V)>\log \frac{pML}{3\ln\frac{1}{1-p}}\right],
\end{align}
where we have used the fact that $e^{-\frac{t}{p}\ln \frac{1}{1-p}} \leq 1 - t$ for $t\in[0,p]$.
\end{proof}
Note that in the stationary memoryless case, $\liminf_{n\to\infty }P_{\sf UV}^{\otimes n}(\mathcal{F}_n)>0$ implies that the covering error (the first term in the supremum in \eqref{e_sm}) is doubly exponentially small.
Therefore in the stationary memoryless case, if the blocklength $n$ is large enough, \eqref{e_sm} continues to hold when the supremum is relaxed to all $\mathcal{F}$.

\subsection{Converse: Theorem~\ref{thm_conv}}\label{sec_simconv}
\begin{proof}
Take $\mathcal{F}:=\left\{\imath_{U;V}(U;V)>\log ML\right\}$.
We have
\begin{align}
&\quad P_{UV}(\mathcal{F})
-\mathbb{P}[\cup_{m=1}^M\cup_{l=1}^L(U_m,V_l)\in\mathcal{F}]
\nonumber\\
&\ge
P_{UV}(\mathcal{F})
-ML(P_U\times P_V)(\mathcal{F})
\label{e133}
\\
&\ge
\sup_{\lambda>0}\{(1-\exp(-\lambda))\mathbb{P}[\imath_{U;V}(U;V)>\log ML+\lambda]\},
\label{e_last}
\end{align}
where \eqref{e_last} follows from \cite[Proposition 13.1]{jingbo2015egamma}, which allows one to upper bound the supremum of \eqref{e133} over $\mathcal{F}$ by the information spectrum.
\end{proof}

\section{Proof of Sub-optimality of Weighted Sampling}\label{app_simple}
Section~\ref{sec_exp} shows the exact error exponent of joint distribution simulation,
where the achievability part relies on a duality with the covering problem,
without giving an explicit construction.
In this section, we consider the performance of a simple \emph{weighted sampler},
which was inspired by a \emph{likelihood encoder} introduced by Yassaee et al.\ \cite{y-isit13} 
in the context of  achievability proofs in network information theory.
Given the observations $U^M=u^M$ and $V^L=v^L$,
the sampler chooses $(\hat{M},\hat{L})=(m,l)$ with probability
\begin{align}
\frac{\exp(\imath_{U;V}(u_m;v_l))}
{\sum_{i,j=1}^{M,L}\exp(\imath_{U;V}(u_i;v_j))}.
\label{e135}
\end{align}
In other words, the probability that a sample is selected is proportional to the ratio of the target distribution to the underlying distribution.
The rationale of such a weighted sampler is reminiscent of,
though not equivalent to,  \emph{importance sampling} in statistics.
The purpose of analyzing such a weighted sampling scheme is two-fold:
\begin{itemize}
\item The weighted sampler is shown to give a strictly suboptimal error exponent.
    This justifies the duality approach in Section~\ref{sec_exp} based on covering lemmas as a singular contribution.
\item The converse proof for the weighted sampler gives another application of
Talagrand's concentration inequality (Theorem~\ref{thm_talagrand}).
\end{itemize}
\subsection{One-shot Achievability Bound}
\begin{thm}\label{thm7}
Fix $P_{UV}$. Let $U_1,\dots,U_M\sim P_U$ and $V_1,\dots,V_L\sim P_V$
be independent.
Let $(U_{\hat{M}},V_{\hat{L}})$ be selected by the weighted sampling scheme.
Then
\begin{align}
&\quad\frac{1}{2}|P_{UV}-P_{U_{\hat{M}}V_{\hat{L}}}|
\nonumber\\
&\le
\mathbb{E}\left[\frac{1}
{1 + ML
\exp(-\imath_{U;V}(U;V))}\right]
\label{e108}
\\
&\le
\inf_{\gamma>0}\left\{\mathbb{P}[\imath_{U;V}(U;V)>\log ML-\gamma]+\exp(-\gamma)\right\},
\label{e109}
\end{align}
where $(U,V)\sim P_{UV}$.
\end{thm}
\begin{proof}
For each $(u,v)$,
define
\begin{align}
Z(u,v):=\exp(\imath_{U;V}(u;v)).
\end{align}
For any event $\mathcal{F}$,
\begin{align}
&\quad P_{U_{\hat{M}}V_{\hat{L}}}(\mathcal{F})
\nonumber\\
&=\mathbb{E}
\left[
\frac{\sum_{m,l=1}^{M,L}
Z(U_m,V_l)1\{(U_m,V_l)\in \mathcal{F}\}
}
{
\sum_{m,l=1}^{M,L}
Z(U_m,V_l)
}
\right]
\nonumber\\
&=ML\, \mathbb{E}\left[
\frac{Z(U_1,V_1)1\{(U_1,V_1)\in\mathcal{F}\}}
{\sum_{m,l=1}^{M,L}
Z(U_m,V_l)}\right]
\label{e_symmetry}
\\
&\ge
ML\,\mathbb{E}\left[
\frac{Z(U_1,V_1)1\{(U_1,V_1)\in\mathcal{F}\}}{Z(U_1,V_1)
+\sum_{(m,l)\neq (1,1)}\mathbb{E}[Z(U_m,V_l)|U_1V_1]}\right]
\label{e_jen}
\\
&=ML\,\mathbb{E}\left[
\frac{Z(U_1,V_1)1\{(U_1,V_1)\in\mathcal{F}\}}{Z(U_1,V_1)+ML-1}\right],
\label{e_113}
\\
&=\mathbb{E}\left[\frac{1\{(U,V)\in\mathcal{F}\}}{1+\frac{1}{ML}Z(U,V)}
\right]
\nonumber
\end{align}
where \eqref{e_symmetry} follows from symmetry,
\eqref{e_jen} follows from Jensen's inequality,
and \eqref{e_113} follows from change of measure.
Then
\begin{align}
&\quad P_{UV}(\mathcal{F})-P_{U_{\hat{M}}V_{\hat{L}}}(\mathcal{F})
\nonumber\\
&\le\mathbb{E}\left[\frac{\frac{1}{ML}\cdot Z(U,V)
1\{(U,V)\in\mathcal{F}\}}{1+\frac{1}{ML}
\cdot Z(U,V)}\right]
\\
&\le \mathbb{P}[\imath_{U;V}(U;V)>\log ML-\gamma]+\exp(-\gamma),
\label{e_infspec}
\end{align}
where \eqref{e_infspec} follows by bounding the quantity in the expectation by 1 when $\imath_{U;V}(U;V)>\log ML-\gamma$, and by $\exp(-\gamma)$ otherwise.
\end{proof}

\subsection{Achievable Error Exponent for Weighted Sampling}
The following result uses Theorem \ref{thm7} to analyze the asymptotic decay of the
total variation distance between the target joint probability measure and that achieved by the 
weighted sampler when each of the $M$ and $L$ observations
are vectors of length $n$.

\begin{thm}\label{thm8}
Fix $P_{\sf UV}$ on $\mathcal{U}\times \mathcal{V}$ and $R_1,R_2>0$ such that $R_1+R_2\ge I({\sf U;V})$.
For each $n$, define
\begin{align}
M_n&=\exp(nR_1);
\\
L_n&=\exp(nR_2).
\end{align}
Let $U^{(n)}_1,\dots,U^{(n)}_{M_n}\sim P_{\sf U}^{\otimes n}$ and $V^{(n)}_1,\dots,V^{(n)}_{L_n}\sim P_{\sf V}^{\otimes n}$ be independent. 
Let $(U^{(n)}_{\hat{M}},V^{(n)}_{\hat{L}})\in \mathcal{U}^n \times \mathcal{V}^n$ be the pair of $n$-tuples selected by the weighted sampling scheme which aims to simulate $P_{\sf UV}^{\otimes n}$.
Then,
\begin{align}
&\quad-\limsup_{n\to\infty} \frac1n \log|P_{U^{(n)}_{\hat{M}}V^{(n)}_{\hat{L}}}-P_{\sf UV}^{\otimes n}|
\nonumber\\
&\ge \max_{\rho\in[0,1]}\rho\left(R_1+R_2-D_{1+\rho}(P_{\sf UV}\|P_{\sf U}\times P_{\sf V})\right).
\label{e-118-1} 
\end{align}
\end{thm}
\begin{proof}
We first observe that the single-shot bound \eqref{e108} can be relaxed to the following
\begin{align}
&\mathbb{E}\left[\frac{1}
{1+{ML}\cdot
\exp(-\imath_{U;V}(U;V))}\right]\\
&\le \inf_{0<\rho\le 1}\mathbb{E}\left[
\left({ML}\cdot
\exp(-\imath_{U;V}(U;V))\right)^{-\rho}\right]
\label{e148}
\\
&=\inf_{0<\rho\le 1}\{(ML)^{-\rho}\exp(\rho  D_{1+\rho}(P_{UV}\|P_U\times P_V))\},\label{eq:n5}
\end{align}
where \eqref{e148} used the simple inequality $\frac{1}{1+t}\le t^{-\rho}$ for $t>0, \rho\in[0,1]$.
The claim then follows by specializing to the stationary memoryless setting $P_{U^{(n)}V^{(n)}}\leftarrow P_{\sf UV}^{\otimes n}$.
\end{proof}
\begin{rem}\normalfont
Define the tilted distribution
\begin{align}
P^{(1+\rho)}_{\sf UV}(u,v) \propto P_{\sf UV}^{1+\rho}(u,v)
P_{\sf U}^{-\rho}(u)P_{\sf V}^{-\rho}(v),
\label{e_prho}
\end{align}
and denote 
\begin{align}
R^{(\rho)} = \mathbb{E} [ \imath_{\sf U;V} ({\sf{U}}^{\rho};{\sf{V}}^{\rho}) ], \quad ({\sf{U}}^{\rho};{\sf{V}}^{\rho}) \sim P_{\sf UV}^{(1+\rho)},
\end{align}
where the information density $ \imath_{\sf U;V}$ is defined with  $P_{\sf UV}$.
 Optimizing the expression \eqref{e-118-1} we obtain the following solution
\begin{itemize}
\item If $0\le R_1+R_2\le R^{(1)}$,
we can express the right side of \eqref{e-118-1} equivalently as 
\begin{align}
\max_{\rho\in[0,1]}\rho\left(R_1+R_2-D_{1+\rho}(P_{\sf UV}\|P_{\sf U}\times P_{\sf V})\right)
=D(P_{\sf UV}^{(1+\rho^\star)}\|P_{\sf UV}),
\label{e152}
\end{align}
where $\rho^\star\in [0,1]$ is the solution to
\begin{align}
R^{(\rho)}=R_1+R_2.
\end{align}
\item If
$R_1+R_2> R^{(1)}$,
\begin{align}
\eqref{e-118-1}&=R_1+R_2-D_{2}(P_{\sf UV}\, \| \,P_{\sf U}\times P_{\sf V})\label{eq:n1}\\
                        &= R_1+R_2+ D(P_{\sf UV}^{(2)}\|P_{\sf UV})-R^{(1)}.
\label{e155}                
\end{align}
\end{itemize}
\end{rem}
\begin{rem}\normalfont
We did not use \eqref{e109} in the proof of Theorem~\ref{thm8} because
the step from \eqref{e108} to \eqref{e109} is not exponentially tight in the stationary memoryless setting.
According to the analysis in Theorem~\ref{thm8}, an exponentially tight approximation of \eqref{e108} is
\begin{align}
\sup_{\gamma>0}
\{\exp(-\gamma)P[\imath_{P\|Q}>\log ML-\gamma]\}.
\label{e116}
\end{align}
In fact, by  the Cramer's large deviation theorem \cite{dembo} in the asymptotic setting, we have the following exponentially tight approximation of \eqref{e116}\footnote{We write $f(n)\doteq g(n)$ if $\frac{1}{n}\log f(n)-\frac{1}{n}\log g(n)\to 0$, $n\to \infty$.}
\begin{align}
&\sup_{\gamma>0}
\{\exp(-\gamma)P[\imath_{\sf U;V}>\log ML-\gamma]\}\nonumber\\
&\doteq \sup_{\gamma>0}
\{\exp(-\gamma)\inf_{\rho>0}\{\exp(\rho\gamma)(ML)^{-\rho}\mathbb{E}[\exp(\rho\imath_{\sf U;V})]\}\}\nonumber\\
&=\exp(\sup_{\gamma>0}\inf_{\rho>0}\{\gamma(\rho-1)-\rho\log(ML)+\log\mathbb{E}[\exp(\rho\imath_{\sf U;V})]\})\nonumber\\
&=\exp(\inf_{\rho>0}\sup_{\gamma>0}\{\gamma(\rho-1)-\rho\log(ML)+\log\mathbb{E}[\exp(\rho\imath_{\sf U;V})]\})\label{sion}\\
&=\exp(\inf_{1\ge\rho>0}\{-\rho\log(ML)+\log\mathbb{E}[\exp(\rho\imath_{\sf U;V})]\}),\label{eq:n6}
\end{align}
where in \eqref{sion} we have used Sion's min-max theorem, since the inside expression is convex in $\rho$ and linear in $\gamma$. Finally \eqref{eq:n6} is \eqref{eq:n5}.
\end{rem}


\subsection{A Converse on the Error Exponent of Weighted Sampling}
\begin{thm}\label{thm9}
Equality in \eqref{e-118-1} (see also \eqref{eq:n1}) holds if $R_1$ and $R_2$ satisfy
\begin{align}
R^{(1)}&<R_1+R_2
\label{e126}
\\
&< 2 R^{(1)} - \max\{D_2(P_{\sf UV}\|P_{\sf U}\times P_{\sf V}),R_1,R_2\},
\label{e127}
\end{align}
where $P^{(2)}_{\sf UV}$  is the tilted distribution defined in \eqref{e_prho}.
\end{thm}
Note that unless $I({\sf U;V})=0$,
we have $R^{(1)}-D_2(P_{\sf UV}\|P_{\sf U}\times P_{\sf V})=D(P_{\sf UV}^{(2)}\|P_{\sf UV})>0$,
and the $(R_1,R_2)$ region in Theorem~\ref{thm9} is nonempty. 
From the analysis in \eqref{e152}-\eqref{e155}, 
we see that $\rho=1$ is achieved in the optimization in \eqref{e-118-1} for $(R_1,R_2)$ in this region,
and hence the exponent in \eqref{e-118-1} is \emph{strictly} worse than \eqref{e31}.
This shows that the weighted sampler yields a \emph{fundamentally} suboptimal exponent (not merely due to a weak achievability proof).

\begin{rem}
In the context of channel coding, the likelihood decoder selects a codeword with probability proportional to the likelihood (analogous to the rule \eqref{e135}).
For rates below the capacity,
a certain ``$\alpha$-likelihood'' decoder achieves the optimal random coding exponent when $\alpha\in[\alpha_{\tt c},\infty]$ where $\alpha_{\tt c}\in[0,1]$ depends on the given rate \cite{scarlett2015likelihood}\cite{merhav2017generalized}\cite{liu2017alpha}.
$\alpha=1$ corresponds to likelihood decoder and $\alpha=\infty$ corresponds to the maximum likelihood decoder. 
However, the situation for the current distribution simulation problem is entirely different: if we select $(m,n)$ according to an ``$\alpha$-likelihood'' rule, i.e.\ with probability proportional to $\exp(\alpha\imath_{U;V}(u_m;v_l))$, then we will not be able to simulate a distribution close to $P_{UV}$ even with infinitely many codewords unless $\alpha=1$, hence we should never choose $\alpha\neq 1$. 
In contrast, in the channel coding problem the error probability is monotonically decreasing in $\alpha$, and $\alpha=\infty$ is optimal \cite{liu2017alpha}.
\end{rem}
\begin{proof}[Proof of Theorem~\ref{thm9}]
To simplify notation, we first consider the single-shot setting with a given joint distribution $P_{UV}$,
and the claim will follow by taking $P_{UV}\leftarrow P_{\sf UV}^{\otimes n}$.

By the same arguments as the proof of Lemma~\ref{le:dsmcl},
we can show the concentration inequality
\begin{align}
&\mathbb{P}\left[\sum_{m=2,l=2}^{M,L}Y_{m,l}\le ML-t\right]
\nonumber\\
&\le
\exp_e\left(
-\tfrac{t^2}{2(ML\exp(D_2(P_{UV}\|P_U\times P_V))+LM(L+M))}
\right),\label{eq:n2}
\end{align}
for any $t\ge 0$, where $Y_{m,l}=\exp(\imath_{U;V}(U_m,V_l))$.

Now by splitting into two events, $\sum_{m=2,l=2}^{M,L}Y_{m,l}\ge ML-t$ and its complement, we have
\begin{align}
&\quad P_{U_{\hat{M}}V_{\hat{L}}}(\mathcal{F})
\nonumber\\
&=ML\mathbb{E}\left[
\frac{\exp(\imath_{U;V}(U_1;V_1))
1\{(U_1,V_1)\in\mathcal{F}\}}{\sum_{m=1,l=1}^{M,L}
\exp(\imath_{U;V}(U_m;V_l))}\right]
\\
&\le ML\mathbb{E}\left[
\frac{\exp(\imath_{U;V}(U_1;V_1))1\{(U_1,V_1)\in\mathcal{F}\}}{
\exp(\imath_{U;V}(U_1;V_1))+
\sum_{m=2,l=2}^{M,L}
\exp(\imath_{U;V}(U_m;V_l))}\right]
\\
&\le
ML\mathbb{E}\left[
\frac{\exp(\imath_{U;V}(U_1;V_1))1\{(U_1,V_1)\in\mathcal{F}\}}
{\exp(\imath_{U;V}(U_1;V_1))+ML-t}
\right]+ML\epsilon
\\
&=ML\mathbb{E}\left[
\frac{1\{(U,V)\in\mathcal{F}\}}
{\exp(\imath_{U;V}(U;V))+ML-t}
\right]+ML\epsilon,
\end{align}
where for brevity $\epsilon$ denotes the right side of \eqref{eq:n2} and  $(U,V)\sim P_{UV}$.
Thus,
\begin{align}
&\quad P_{UV}(\mathcal{F})-P_{U_{\hat{M}}V_{\hat{L}}}(\mathcal{F}) 
\nonumber
\\
&\ge
\mathbb{E}\left[
\frac{\exp(\imath_{U;V}(U;V))-t}
{\exp(\imath_{U;V}(U;V))+ML-t}
\cdot
1\{(U,V)\in\mathcal{F}\}
\right]
-ML\epsilon.
\label{e_fraction}
\end{align}
Note that the right side is not necessarily maximized when $\mathcal{F}$ is the whole set since the numerator can be negative.
Let $\gamma:= \log ML-R^{(2)}\ge 0$, which is exponentially equivalent to the optimal $\gamma$ in
\eqref{e116}.
Set
\begin{align}
\mathcal{F}:=\{\imath_{U;V}\ge \log ML-\gamma\}.
\end{align}
By choosing
\begin{align}
t=\frac{1-\exp(-\gamma)}{2-\exp(-\gamma)}\cdot ML\exp(-\gamma),
\label{e168}
\end{align}
we can ensure that the fraction in the expectation in the right side of \eqref{e_fraction} is at least $\exp(-\gamma)/2$ for $(U,V)\in\mathcal{F}$.
Indeed, by monotonicity it suffices to verify by substituting $\imath_{U;V}(U;V)$ with $\log ML-\gamma$, 
in which case the choice of $t$ in \eqref{e168} ensures that the fraction equals exactly $\exp(-\gamma)/2$.
\begin{align}
\frac{1}{2}|P-\hat{P}|
\ge
\frac{1}{2}\exp(-\gamma) \mathbb{P} [\imath_{\sf U;V} ({\sf U;V}) >\log ML-\gamma]
-ML\epsilon.
\label{e139}
\end{align}
Finally, consider the asymptotic setting where we take $P_{UV}\leftarrow P_{\sf UV}^{\otimes n}$,
$M\leftarrow \exp(nR_1)$, and $L\leftarrow \exp(nR_2)$,
then
\begin{align}
t^2&\doteq M^2L^2\exp(-2\gamma)=\exp \left(2 R^{(2)}\right);
\end{align}
Therefore $\epsilon$ vanishes doubly exponentially under the assumptions \eqref{e126} and \eqref{e127}.
The right side of \eqref{e139} is thus exponentially equivalent to the right side of \eqref{e116}.
\end{proof}

\section{Proof of the Multivariate Extension}\label{apx:multivariate}
This section proves Lemma \ref{le:multivariate}. 
The idea is to analyze the following weighted sum for any instance $z$ using the Talagrand inequality and then taking the typical with respect to $P_Z$,
\begin{align*}
S(z):=\dfrac{1}{M_1\cdots M_k}\sum \exp&(\imath(z,V_{1,m_1},\cdots,V_{k,m_k}))\nonumber\\&1\{(z,V_{1,m_1},\cdots,V_{k,m_k})\in\mathcal{G}\}.
\end{align*}
where $\imath(z,v_1,\cdots,v_k):=\imath_{P_{ZV^k}||P_Z\prod_{i=1}^kP_{V_i|Z}}(z,v_1,\cdots,v_k)$.
We proceed with the same steps as in the proof of Lemma~\ref{le:dsmcl}. First, observe that the mean of $S(z)$ is
\begin{equation}
\mathbb{E}[S(z)]=\mathbb{P}[(Z,V_1,\cdots,V_k)\in\mathcal{G}|Z=z].
\end{equation}
Define
\begin{align*}Y_{m_1,\cdots,m_k}:=&\dfrac{1}{M_1\cdots M_k} \exp(\imath(z,V_{1,m_1},\cdots,V_{k,m_k}))\nonumber\\&1\{(z,V_{1,m_1},\cdots,V_{k,m_k})\in\mathcal{G}\}.\end{align*}
Next in Theorem \ref{thm_talagrand}, let $f(V_1^{M_1},\cdots,V_k^{M_k})=S(z)$ and
\begin{align}
&c_{i,m_i}(V_1^{M_1},\cdots,V_k^{M_k})
:=\sum_{m_1,\cdots,m_{i-1},m_{i+1},\cdots,m_k} Y_{m_1,\cdots,m_k},\nonumber\\&\quad \forall m_i\in \{1,\dots,M_i\}
.
\end{align}
Then, the one-sided Lipshitz property \eqref{theTcondition} is satisfied since
\begin{align}
f(v_1^{M_1}&,\cdots, v_k^{M_k})-f(u_1^{M_1},\cdots, u_k^{M_k})\nonumber\\
&\le \sum_{i=1}^k\sum_{m_i=1}^{M_i} c_{i,m_i}(v_1^{M_1}, \cdots,v_k^{M_k}) 1\{v_{i,m_i}\neq u_{i,m_i}\}.\nonumber
\end{align}
The variance proxy in Talagrand's inequality (i.e. the denominator in the \eqref{e_talagrand}) is
\begin{align}
&\quad\sum_{i=1}^k\sum_{m_i=1}^{M_i}\mathbb{E}[c_{i,m_i}^2(V_1^{M_1},\cdots,V_k^{M_k})]\nonumber\\
&=\sum_{i=1}^kM_i\mathbb{E}[c_{i,1}^2(V_1^{M_1},\cdots,V_k^{M_k})],\label{eq:mul-00}
\end{align}
where the equality follows from symmetry.

Next we simplify the term inside the summation for $i=1$, (the simplification for $i>1$ is similar)
\begin{align}
&M_1\mathbb{E}[c_{1,1}^2(V_1^{M_1},\cdots,V_k^{M_k})]\nonumber\\
&=\sum_{m_2,\cdots,m_k}\sum_{m'_2,\cdots,m'_k}\mathbb{E}[Y_{1,m_2,\cdots,m_k}Y_{1,m'_2,\cdots,m'_k}]\nonumber\\
&=M_1\cdots M_k\sum_{m_2,\cdots,m_k}\mathbb{E}[Y_{1,1,\cdots,1}Y_{1,m_2,\cdots,m_k}]\label{eq:sym-mul}\\
&=M_1\cdots M_k\sum_{\mathcal{S}\subseteq[2:k]}\sum_{m_2,\cdots,m_k:\atop{ m_t=1, t\in\mathcal{S}\atop m_t\neq 1, t\notin\mathcal{S}}}\mathbb{E}[Y_{1,1,\cdots,1}Y_{1,m_2,\cdots,m_k}],\label{eq:mul-10}
\end{align}
where \eqref{eq:sym-mul} is again follows from symmetry.

Next for each $(m_1,\cdots,m_k)$ satisfying $m_t=1, t\in\mathcal{S}$ and $m_t\neq 1, t\notin\mathcal{S}$, we have
\begin{align}
&\mathbb{E}[Y_{1,1,\cdots,1}Y_{1,m_2,\cdots,m_k}]\le\dfrac{1}{(M_1\cdots M_k)^2}\mathbb{E}[\nonumber\\& \exp(\imath(z,V_{1,1},\cdots,V_{k,1})+\imath(z,V_{1,m_1},\cdots,V_{k,m_k}))\nonumber\\&\qquad1\{(z,V_{1,1},\cdots,V_{k,1})\in\mathcal{G}\}]\\
&=\dfrac{1}{(M_1\cdots M_k)^2}\mathbb{E}_{V_{1,1},\cdots,V_{k,1}}[\nonumber\\& \exp(\imath(z,V_{1,1},\cdots,V_{k,1})+\imath_{P_{ZV_{\mathcal{S}}}||P_Z\prod_{i\in\mathcal{S}}P_{V_{i}|Z}}(z,\{V_{i,1}:i\in\mathcal{S}\}))\nonumber\\&\qquad1\{(z,V_{1,1},\cdots,V_{k,1})\in\mathcal{G}\}\mathbb{E}_{\{V_{i,m_i}:i\notin\mathcal{S}\}}[\nonumber\\
&\qquad\exp(\imath_{P_{V_{\mathcal{S}^c}|Z=z,V_{\mathcal{S}}=V_{\mathcal{S},1}}||\prod_{i\notin\mathcal{S}}P_{V_{i}|Z}}(z,\{V_{i,m_i}:i\notin\mathcal{S}\}))]]
\label{eq:mult-2}\\
&=\dfrac{1}{(M_1\cdots M_k)^2}\mathbb{E}_{V_{1,1},\cdots,V_{k,1}}[\nonumber\\& \exp(\imath(z,V_{1,1},\cdots,V_{k,1})+\imath_{P_{ZV_{\mathcal{S}}}||P_Z\prod_{i\in\mathcal{S}}P_{{V_i}|Z}}(z,\{V_{i,1}:i\in\mathcal{S}\}))\nonumber\\&\qquad1\{(z,V_{1,1},\cdots,V_{k,1})\in\mathcal{G}\}]
\label{eq:mult-3}\\
&\le\dfrac{\exp(-\gamma)}{M_1\cdots M_k\prod_{t\notin\mathcal{S}} M_t}\mathbb{E}[ \exp(\imath(z,V_{1,1},\cdots,V_{k,1}))\nonumber\\&\qquad1\{(z,V_{1,1},\cdots,V_{k,1})\in\mathcal{G}\}]\label{eq:mult-4}\\
&=\dfrac{\exp(-\gamma)}{M_1\cdots M_k\prod_{t\notin\mathcal{S}} M_t}\mathbb{P}[ (Z,V_{1},\cdots,V_{k})\in\mathcal{G}|Z=z],\label{eq:mul-20}
\end{align}
where
\begin{itemize}
\item \eqref{eq:mult-2} follows from the chain rule for information density. 
\item \eqref{eq:mult-3} follows from change of measure and the fact that conditioned on $Z=z$, $\{V_{i,m_i}:i\notin\mathcal{S}\}\sim \prod_{i\notin\mathcal{S}}P_{V_i|Z=z}$.
\item \eqref{eq:mult-4} follows from the definition of $\mathcal{G}$ in \eqref{eq:dfn-G}.
\item \eqref{eq:mul-20} follows from change of measure and the definition of $\imath(z,v_1,\cdots,v_k)$ in \eqref{eq:dfn-I}.
\end{itemize}

Invoking \eqref{eq:mul-10} and \eqref{eq:mul-20}, we get
\begin{align}
&M_1\mathbb{E}[c_{1,1}^2(V_1^{M_1},\cdots,V_k^{M_k})]\nonumber\\
&\le\exp(-\gamma)\mathbb{P}[ (Z,V_{1},\cdots,V_{k})\in\mathcal{G}|Z=z]\nonumber\\&\qquad
\sum_{\mathcal{S}\subseteq[2:k]}\sum_{m_2,\cdots,m_k:\atop{ m_t=1, t\in\mathcal{S}\atop m_t\neq 1, t\notin\mathcal{S}}}\dfrac{1}{\prod_{t\notin\mathcal{S}} M_t}\\
&<2^{k-1}\exp(-\gamma)\mathbb{P}[ (Z,V_{1},\cdots,V_{k})\in\mathcal{G}|Z=z].
\end{align}
Consequently, the variance proxy \eqref{eq:mul-00} is upper bounded by
\[k2^{k-1}\exp(-\gamma)\mathbb{P}[ (Z,V_{1},\cdots,V_{k})\in\mathcal{G}|Z=z].\]

Finally, taking $t=\mathbb{E}[f(V_1^{M_1},\cdots,V_{k}^{M_k})]$ in  Theorem~\ref{thm_talagrand}, we obtain
\begin{align*}
\mathbb{P}\left[S(z)\right.&\left.=0\right]
= \mathbb{P}[f(V_1^{M_1},\cdots,V_{k}^{M_k})\le 0]
\\
&= \mathbb{P}[f(V_1^{M_1},\cdots,V_{k}^{M_k})\le \mathbb{E}[f(V_1^{M_1},\cdots,V_{k}^{M_k})]-t]
\\
&\le \exp_{e}\left(
-\frac{\mathbb{P}[ (Z,V_{1},\cdots,V_{k})\in\mathcal{G}|Z=z]}{k2^k\exp(-\gamma)
}
\right).
\end{align*}

\section{Conclusion}
We derived a sharp mutual covering lemma (Lemma~\ref{le:dsmcl}) and demonstrated its tightness in two regimes which we called ``typical'' and ``worst case''.
Our main result is presented in the single-shot form for simplicity and elegance, 
but the asymptotic sharpness really comes from the tailor-made concentration inequality and the novel duality argument in the proof.
We showed two applications where the two regimes are the most useful, 
respectively: broadcast channel and joint distribution simulation.

\section{Acknowledgement}
This work was supported in part by
NSF Grants CCF-1016625, CCF-0939370, and by the Center for Science of Information.
Jingbo Liu's work was supported by the Wallace Memorial Fellowship at Princeton University.


\bibliographystyle{IEEEtran}
\bibliography{MC}

\begin{thebibliography}{10}
\providecommand{\url}[1]{#1}
\csname url@samestyle\endcsname
\providecommand{\newblock}{\relax}
\providecommand{\bibinfo}[2]{#2}
\providecommand{\BIBentrySTDinterwordspacing}{\spaceskip=0pt\relax}
\providecommand{\BIBentryALTinterwordstretchfactor}{4}
\providecommand{\BIBentryALTinterwordspacing}{\spaceskip=\fontdimen2\font plus
\BIBentryALTinterwordstretchfactor\fontdimen3\font minus
  \fontdimen4\font\relax}
\providecommand{\BIBforeignlanguage}[2]{{%
\expandafter\ifx\csname l@#1\endcsname\relax
\typeout{** WARNING: IEEEtran.bst: No hyphenation pattern has been}%
\typeout{** loaded for the language `#1'. Using the pattern for}%
\typeout{** the default language instead.}%
\else
\language=\csname l@#1\endcsname
\fi
#2}}
\providecommand{\BIBdecl}{\relax}
\BIBdecl

\bibitem{csiszar}
I.~Csisz\'ar and J.~K\"orner, \emph{Information Theory: Coding Theorems for
  Discrete Memoryless Systems}.\hskip 1em plus 0.5em minus 0.4em\relax
  Cambridge University Press, 2011.

\bibitem{NIT}
A.~E. Gamal and Y.~H. Kim, \emph{Network Information Theory}.\hskip 1em plus
  0.5em minus 0.4em\relax Cambridge University Press, 2011.

\bibitem{Verduallerton}
S.~Verd\'u, ``Non-asymptotic achievability bounds in multiuser information
  theory,'' in \emph{Fiftieth Annual Allerton Conference on Communication,
  Control and Computing}, Monticello, Illinois, Oct. 2012, pp. 1--8.

\bibitem{el1981proof}
A.~E. Gamal and E.~C. van~der Meulen, ``A proof of {Marton's} coding theorem
  for the discrete memoryless broadcast channel,'' \emph{IEEE Trans.\ Inf.\
  Theory}, vol.~27, no.~1, pp. 120--122, Jan.~1981.

\bibitem{alon2004probabilistic}
N.~Alon and J.~H. Spencer, \emph{The Probabilistic Method}.\hskip 1em plus
  0.5em minus 0.4em\relax John Wiley \& Sons, 2004.

\bibitem{jingbo}
J.~Liu, P.~Cuff, and S.~Verd\'u, ``One-shot mutual covering lemma and
  {Marton}'s inner bound with a common message,'' in \emph{Proc. 2015 IEEE
  Symp. Inf. Theory (ISIT)}, June 2015, pp. 1457--1461.

\bibitem{lv18}
J.~Liu and S.~Verd\'u, ``Rejection sampling and noncausal sampling under moment
  constraints,'' \emph{Proc. ISIT 2018}, June 2018.

\bibitem{y-isit15}
M.~H. Yassaee, ``One-shot achievability via fidelity,'' in \emph{Proc. 2015
  IEEE Symp. Inf. Theory}, June 2015, pp. 301--305.

\bibitem{song16}
E.~C. Song, P.~Cuff, and H.~V. Poor, ``The likelihood encoder for lossy
  compression,'' \emph{IEEE Trans.\ Inf.\ Theory}, vol.~62, pp. 1836--1849,
  April 2016.

\bibitem{watanabe}
S.~Watanabe, S.~Kuzuoka, and V.~Y.~F. Tan, ``Non-asymptotic and second-order
  achievability bounds for coding with side-information,'' \emph{IEEE Trans.
  Inf. Theory}, vol.~61, no.~4, pp. 1574--1605, Apr.~2015.

\bibitem{liu2016key}
J.~Liu, P.~Cuff, and S.~Verd{\'u}, ``Key capacity for product sources with
  application to stationary gaussian processes,'' \emph{IEEE Transactions on
  Information Theory}, vol.~62, no.~2, pp. 984--1005, 2016.

\bibitem{jingbo2015egamma}
J.~Liu, P.~Cuff, and S.~Verd\'u, ``{$E_{\gamma}$}-resolvability,'' \emph{IEEE
  Trans. Inf. Theory}, vol.~63, pp. 2629--2658, May 2017.

\bibitem{liu2017secret}
J.~Liu, P.~Cuff, and S.~Verd{\'u}, ``Secret key generation with limited
  interaction,'' \emph{IEEE Transactions on Information Theory}, vol.~63,
  no.~11, pp. 7358--7381, 2017.

\bibitem{harsha2010}
P.~Harsha, R.~Jain, D.~McAllester, and J.~Radhakrishnan, ``The communication
  complexity of correlation,'' \emph{IEEE Trans. Inf. Theory}, vol.~56, no.~1,
  pp. 438--449, Jan.~2010.

\bibitem{LiElGamal}
C.~T. Li and A.~E. Gamal, ``Strong functional representation lemma and
  applications to coding theorems,'' \emph{arXiv preprint arXiv:1701.02827},
  2017.

\bibitem{li2018unified}
C.~T. Li and V.~Anantharam, ``A unified framework for one-shot achievability
  via the poisson matching lemma,'' \emph{arXiv preprint arXiv:1812.03616},
  2018.

\bibitem{witsenhausen1975}
H.~S. Witsenhausen, ``On sequences of pairs of dependent random variables,''
  \emph{SIAM Journal on Applied Mathematics}, vol.~28, no.~1, pp. 100--113,
  Jan. 1975.

\bibitem{mossel2006}
E.~Mossel, R.~O'Donnell, O.~Regev, J.~Steif, and B.~Sudakov, ``Non-interactive
  correlation distillation, inhomogeneous markov chains, and the reverse
  {Bonami-Beckner} inequality,'' \emph{Israel Journal of Mathematics}, vol.
  154, no.~1, pp. 299--336, 2006.

\bibitem{kamath}
S.~Kamath and V.~Anantharam, ``On non-interactive simulation of joint
  distributions,'' \emph{IEEE Trans.\ Inf.\ Theory}, vol.~62, pp. 3419--3435,
  June 2016.

\bibitem{Radhakrishnan2016}
J.~Radhakrishnan, P.~Sen, and N.~Warsi, ``One-shot {Marton} inner bound for
  classical-quantum broadcast channel,'' \emph{IEEE Trans. Inf. Theory},
  vol.~62, no.~5, pp. 2836--2848, May 2016.

\bibitem{BLM}
S.~Boucheron, G.~Lugosi, and P.~Massart, \emph{Concentration Inequalities: A
  Nonasymptotic Theory of Independence}.\hskip 1em plus 0.5em minus 0.4em\relax
  Oxford University Press, 2013.

\bibitem{verdu-jerusalem}
S.~Verd\'{u}, ``Non-asymptotic covering lemmas,'' in \emph{Proc. 2015 IEEE
  Information Theory Workshop}, Jerusalem, Israel, April 26--May 1, 2015.

\bibitem{renner2005simple}
R.~Renner and S.~Wolf, ``Simple and tight bounds for information reconciliation
  and privacy amplification,'' in \emph{Advances in Cryptology-ASIACRYPT
  2005}.\hskip 1em plus 0.5em minus 0.4em\relax Springer, 2005, pp. 199--216.

\bibitem{y-isit13}
M.~H. Yassaee, M.~R. Aref, and A.~Gohari, ``A technique for deriving one-shot
  achievability results in network information theory,'' in \emph{Proc. 2013
  IEEE Symp. Inf. Theory (ISIT)}, July 2013, pp. 1287--1291.

\bibitem{scarlett2015}
J.~Scarlett, ``On the dispersions of the {Gel'fand--Pinsker} channel and dirty
  paper coding,'' \emph{IEEE Transactions on Information Theory}, vol.~61,
  no.~9, pp. 4569--4586, 2015.

\bibitem{merhav}
N.~Merhav, ``Exact random coding error exponents of optimal bin index
  decoding,'' \emph{IEEE Trans. Inf. Theory}, vol.~60, no.~10, pp. 6024--6031,
  August 2014.

\bibitem{yagli2018exact}
S.~Yagli and P.~Cuff, ``Exact soft-covering exponent,'' in \emph{2018 IEEE
  International Symposium on Information Theory (ISIT)}.\hskip 1em plus 0.5em
  minus 0.4em\relax IEEE, 2018, pp. 1680--1684.

\bibitem{marton1979}
K.~Marton, ``A coding theorem for the discrete memoryless broadcast channel,''
  \emph{IEEE Trans. Inf. Theory}, vol.~25, no.~3, pp. 306--311, 1979.

\bibitem{liangK2007}
Y.~Liang and G.~Kramer, ``Rate regions for relay broadcast channels,''
  \emph{IEEE Transactions on Information Theory}, vol.~53, no.~10, pp.
  3517--3535, 2007.

\bibitem{gallager}
R.~Gallager, \emph{Information Theory and Reliable Communication}.\hskip 1em
  plus 0.5em minus 0.4em\relax Wiley 1968.

\bibitem{combinatorics}
R.~A. Brualdi, \emph{Introductory Combinatorics}, 5th~ed.\hskip 1em plus 0.5em
  minus 0.4em\relax Englewood Cliffs: Prentice Hall, 2009.

\bibitem{book}
S.~Verd\'u, \emph{Information Theory}.\hskip 1em plus 0.5em minus 0.4em\relax
  in preparation.

\bibitem{talagrand1988}
M.~Talagrand, ``An isoperimetric theorem on the cube and the kintchine-kahane
  inequalities,'' \emph{Proc. Amer. Math. Soc.}, vol. 104, no.~3, pp. 905--909,
  1988.

\bibitem{talagrand1995}
------, ``Concentration of measure and isoperimetric inequalities in product
  spaces,'' \emph{Inst. Hautes \'Etudes Sci. Publ. Math}, vol.~81, pp. 73--205,
  1995.

\bibitem{raginsky2014concentration}
M.~Raginsky and I.~Sason, ``Concentration of measure inequalities in
  information theory, communications, and coding,'' \emph{Foundations and
  Trends in Communications and Information Theory}, vol.~10, no. 1--2, pp.
  1--250, 2013 (revision: 2014).

\bibitem{pradhan}
D.~Krithivasan and S.~S. Pradhan, ``On large deviation analysis of sampling
  from typical sets,'' in \emph{Proc. Workshop on Information theory and
  Applications (ITA)}, 2007.

\bibitem{suen1990}
W.~Suen, ``A correlation inequality and {Poisson} limit theorem for
  nonoverlapping balanced subgraphs of a random graph,'' \emph{Random
  Structures Algorithms}, vol.~1, no. 231-242, 1990.

\bibitem{janson1998}
S.~Janson, ``New versions of {Suen's} correlation inequality,'' \emph{Random
  Structures Algorithms}, vol.~13, pp. 467--483, 1998.

\bibitem{hall}
P.~Hall, ``On representatives of subsets,'' \emph{J. London Math. Soc.},
  vol.~10, no.~1, pp. 26--30, 1935.

\bibitem{dembo}
A.~Dembo and O.~Zeitouni, \emph{Large Deviations Techniques and Applications},
  ser. Applications of mathematics.\hskip 1em plus 0.5em minus 0.4em\relax
  Springer, 1998.

\bibitem{scarlett2015likelihood}
J.~Scarlett, A.~Martinez, and A.~G. i~F{\`a}bregas, ``The likelihood decoder:
  error exponents and mismatch,'' in \emph{2015 IEEE International Symposium on
  Information Theory (ISIT)}.\hskip 1em plus 0.5em minus 0.4em\relax IEEE,
  2015, pp. 86--90.

\bibitem{merhav2017generalized}
N.~Merhav, ``The generalized stochastic likelihood decoder: random coding and
  expurgated bounds,'' \emph{IEEE Transactions on Information Theory}, vol.~63,
  no.~8, pp. 5039--5051, 2017.

\bibitem{liu2017alpha}
J.~Liu, P.~Cuff, and S.~Verd{\'u}, ``On $\alpha$-decodability and
  $\alpha$-likelihood decoder,'' in \emph{2017 55th Annual Allerton Conference
  on Communication, Control, and Computing (Allerton)}.\hskip 1em plus 0.5em
  minus 0.4em\relax IEEE, 2017, pp. 118--124.

\end{thebibliography}

\end{document}